\documentclass[12pt,reqno]{amsart}
\usepackage{amsmath}
\usepackage{amssymb}
\usepackage{verbatim}
\usepackage{mathrsfs}

\newcommand{\sgn}{\operatorname{sgn}}

\usepackage[top=0.5in,bottom=0.5in,left=0.7in,right=0.7in]{geometry}

\newtheorem{lemma}{Lemma}
\newtheorem{prop}[lemma]{Proposition}
\newtheorem{cor}[lemma]{Corollary}
\newtheorem{theorem}[lemma]{Theorem}

\newtheorem{definition}[lemma]{Definition}

\numberwithin{equation}{section}
\numberwithin{lemma}{section}
\newcommand{\C}{\mathbb{C}}    %complex number field
\newcommand{\N}{\mathbb{N}}    %natural numbers
\newcommand{\PL}{\mathbb{P}}   %ring of all polynomials
\newcommand{\R}{\mathbb{R}}    %real number field
\newcommand{\Z}{\mathbb{Z}}    %integers

\newcommand{\wh}{\widehat}
\renewcommand{\le}{\leqslant}
\renewcommand{\ge}{\geqslant}
\newcommand{\bs}{\backslash}
\newcommand{\ol}{\overline}
\newcommand{\la}{\langle}
\newcommand{\ra}{\rangle}
\newcommand{\bo}{\mathscr{O}} %standard big O notation
\newcommand{\gep}{\varepsilon}
\newcommand{\gl}{\lambda}

\newcommand{\bp}{ \begin{proof} }
\newcommand{\ep}{\hfill \end{proof} }
\newcommand{\be}{ \begin{equation} }
\newcommand{\ee}{ \end{equation} }
\newcommand{\tp}{\mathsf{T}}  %transpose
\newcommand{\setsp}{\;:\;}     %set separator
\newcommand{\qp}{\mathcal{Q}}  %quasi-projection operator
\newcommand{\pp}{\mathsf{p}}
\newcommand{\pq}{\mathsf{q}}
\newcommand{\td}{\boldsymbol{\delta}}  %Dirac/Kronicker sequence
\newcommand{\si}{\mathtt{S}}            %shift invariance space
\newcommand{\vmo}{\operatorname{vm}} %\vm=vec{m}. Vanishing moment
\newcommand{\lp}[1]{l_{#1}(\mathbb{Z})}
\newcommand{\Lp}[1]{L_{#1}(\mathbb{R})}
\newcommand{\lLp}[1]{L_{#1}^{loc}(\mathbb{R})}
\newcommand{\TLp}[1]{L_{#1}(\mathbb{T})}
\newcommand{\AS}{\operatorname{\mathsf{AS}}} %time domain affine system

\begin{document}

\title{Gibbs Phenomenon of Framelet Expansions and Quasi-projection Approximation}

\author{Bin Han}

\address{Department of Mathematical and Statistical Sciences,
University of Alberta, Edmonton,\quad Alberta, Canada T6G 2G1.
\quad {\tt bhan@ualberta.ca}\quad {\tt http://www.ualberta.ca/$\sim$bhan}
}

%\author{Xiaosheng Zhuang}

%\address{Department of Mathematics, City University of Hong Kong, Tat Chee Avenue, Kowloon Tong, Hong Kong.\newline{\tt xzhuang7@cityu.edu.hk}}

\thanks{Research was supported in part by the Natural Sciences and Engineering Research Council of Canada (NSERC).
}

%Research of Xiaosheng Zhuang was  supported by the Research Grants Council of Hong Kong (Project No. CityU ???).}

%\thanks{Contact information of corresponding author Bin Han: E-mail: bhan@ualberta.ca, Phone: 1-587-8828375, Fax: 1-780-4926826,  Web: http://www.ualberta.ca/$\sim$bhan}

\makeatletter \@addtoreset{equation}{section} \makeatother

\begin{abstract}
The Gibbs phenomenon is widely known for Fourier expansions of periodic functions and refers to the phenomenon that
the $n$th Fourier partial sums overshoot a target function at jump discontinuities in such a way that such overshoots do not die out as $n$ goes to infinity.
The Gibbs phenomenon for wavelet expansions using (bi)orthogonal wavelets has been studied in the literature.
Framelets (also called wavelet frames) generalize (bi)orthogonal wavelets.
Approximation by quasi-projection operators are intrinsically linked to approximation by truncated wavelet and framelet expansions. In this paper we shall establish a key identity for quasi-projection operators and then we use it to study the Gibbs phenomenon of framelet expansions and approximation by general quasi-projection operators. We shall also study and characterize the Gibbs phenomenon at an arbitrary point for approximation by quasi-projection operators.
As a consequence, we show that the Gibbs phenomenon appears at all points for every tight or dual framelet having at least two vanishing moments and for quasi-projection operators having at least three accuracy orders.
Our results not only improve current results in the literature on the Gibbs phenomenon for (bi)orthogonal wavelet expansions but also are new for framelet expansions and approximation by quasi-projection operators.
\end{abstract}


\maketitle

\pagenumbering{arabic}

\section{Introduction}
Wavelets and their generalizations such as framelets have been widely applied to many areas such as image processing and numerical algorithms with great success (\cite{daubook}). It is well noticed that many wavelets and framelets suffer the visually unpleasant ringing effect near jump discontinuities, which is related to the Gibbs phenomenon of wavelet and framelet expansions. It is the purpose of this paper to study the Gibbs phenomenon of wavelet and framelet expansions as well as their associated approximation schemes by quasi-projection operators.
%Our results in this paper in return will be helpful for finding possible ways to avoid the the Gibbs phenomenon for wavelets expansions and to reduce their ring effects.

Let us recall the Gibbs phenomenon for Fourier expansions.
Let $\TLp{2}$ denote the space of all $2\pi$-periodic square integrable functions equipped with the inner product
\[
\la f, g\ra_{\TLp{2}}:=\frac{1}{2\pi}\int_{-\pi}^\pi f(x) \ol{g(x)}dx,\qquad f,g\in \TLp{2}.
\]
Since $\{e^{ikx}\}_{k\in \Z}$ is an orthonormal basis of $\TLp{2}$, every $f\in \TLp{2}$ has a Fourier expansion $f=\sum_{k\in \Z} \wh{f}(k) e^{ikx}$ in $\TLp{2}$ with Fourier coefficients
\[
\wh{f}(k):=\la f, e^{ikx}\ra=
\frac{1}{2\pi}\int_{-\pi}^\pi f(x) e^{-ikx}dx,\qquad k\in \Z.
\]
That is, $\lim_{n\to \infty}\| S_n f-f\|_{\TLp{2}}=0$, where its $n$th Fourier partial sum $S_n f$ is defined to be
\[
[S_n f](x):=\sum_{k=-n}^n \wh{f}(k) e^{ik x}, \qquad x\in \R.
\]
In applications, a lot of signals are modeled by piecewise smooth/analytic functions with finitely many simple jump discontinuities. Let us consider a particular function $f\in \TLp{2}$ defined by $f(x)=-1$ for $x\in (-\pi,0]$ and $f(x)=1$ for $x\in (0,\pi]$.
Then $f$ is a piecewise smooth/analytic $2\pi$-periodic function with simple jump discontinuities at $\pi\Z$.
By a straightforward calculation,
one has its $n$th Fourier partial sum
\[
[S_nf](x)=\sum_{k=1}^n \frac{2(1-(-1)^k)}{\pi k} \sin (kx).
\]
At every $x_0\not\in \pi \Z$ (i.e., at every point where $f$ is continuous),  $\lim_{n\to \infty} [S_n f](x)=f(x)$ uniformly for $x$ in some neighborhood of $x_0$. But Gibbs \cite{gib} pointed out that at the jump discontinuity $0$,
\[
\lim_{n\to \infty} [S_n f]\left(\frac{\pi}{n}\right)=\frac{2}{\pi}\int_0^\pi \frac{\sin x}{x}dx \approx 1.17898 >1.
\]
That is, $S_n f$ overshoots $f$ by a fixed positive amount from the right-hand side of the origin.
%Since both $f$ and $f_n$ are antisymmetric about the origin, $S_nf$ undershoots $f$ from the left side of the origin.
Such overshoots do not die out as $n$ goes to infinity and this creates undesired visually unpleasant ringing effects near jump discontinuities in applications. This phenomenon is called the \emph{Gibbs phenomenon}, which was first discovered analytically by Wilbraham~\cite{wil48}.
There are a lot of studies on the Gibbs phenomenon and how to reduce it for various expansions and bases.
See \cite{gs97,jerbook} for historical overview and recent developments on the Gibbs phenomenon for Fourier expansions.
The Gibbs phenomenon for orthogonal wavelets and biorthogonal wavelets has been studied in the literature, e.g., see \cite{gc95,jerbook,kel96,ml18,rf05,shen11,shen02,sv96} and references therein.
In this paper we are particularly interested in the Gibbs phenomenon for framelet expansions and their associated quasi-projection approximation.

Every square integrable function in $\Lp{2}$ has a wavelet expansion.
Let us first recall the definition of wavelets and framelets (e.g., see \cite{daubook,han97,hanbook,rs97} and references therein).
For a (vector) function $f$ on the real line $\R$, we shall adopt the notation
\[
f_{j;k}(x):=2^{j/2} f(2^j x-k),\qquad j,k\in \Z, x\in \R.
\]
By $f\in (\Lp{2})^{r\times s}$ we mean that $f$ is an $r\times s$ matrix of functions in $\Lp{2}$ and we define
\be \label{innerprod}
\la f,g\ra:=\int_{\R} f(x) \ol{g(x)}^\tp dx,\qquad
f\in (\Lp{2})^{r\times t}, g\in (\Lp{2})^{s\times t}.
\ee
%
%which an $r\times s$ matrix of complex numbers.
For an $r\times 1$ vector function $\phi=(\phi^1,\ldots,\phi^r)^\tp \in (\Lp{2})^r$ and an $s\times 1$ vector function $\psi=(\psi^1,\ldots,\psi^s)^\tp\in (\Lp{2})^s$, we say that $\{\phi;\psi\}$ is \emph{a framelet in $\Lp{2}$} if there exist positive constants $C_1$ and $C_2$ such that
\be \label{framelet}
C_1 \|f\|_{\Lp{2}}^2\le
\sum_{\ell=1}^r \sum_{k\in \Z} |\la f, \phi^\ell(\cdot-k)\ra|^2+
\sum_{j=0}^\infty \sum_{\ell=1}^s \sum_{k\in \Z} |\la f, \psi^\ell_{j;k}\ra|^2\le C_2 \|f\|_{\Lp{2}}^2,\qquad
\forall\, f\in \Lp{2}.
\ee
%
%Note that $\|\la f, \phi(\cdot-k)\ra\|_{l_2}^2=
%\sum_{\ell=1}^r |\la f, \phi^\ell(\cdot-k)\ra|^2$ and $\|\la f, \psi_{j;k}\ra\|_{l_2}^2=\sum_{\ell=1}^s |\la f, \psi^\ell_{j;k}\ra|^2$.
If \eqref{framelet} holds with $C_1=C_2=1$, then
$\{\phi;\psi\}$ is called \emph{a tight framelet} (or a tight multiframelet) in $\Lp{2}$.
For $J\in \Z$, we define \emph{nonhomogeneous affine systems} $\AS_J(\phi;\psi)$ in $\Lp{2}$ to be
\be \label{as}
\AS_J(\phi;\psi):=\{\phi^\ell_{J;k} \setsp k\in \Z, \ell=1,\ldots,r\}\cup \{ \psi^\ell_{j;k} \setsp j\ge J,k\in \Z,\ell=1,\ldots,s\}.
\ee
If $\AS_0(\phi;\psi)$ is an orthonormal basis of $\Lp{2}$, then $\{\phi;\psi\}$ is called \emph{an orthogonal wavelet} or more precisely, an orthogonal multiwavelet. Obviously, an orthogonal wavelet must be a tight framelet.
Let $\tilde{\phi}$ be an $r\times 1$ vector of functions in $\Lp{2}$ and $\tilde{\psi}$ be an $s\times 1$ vector of functions in $\Lp{2}$. We say that $(\{\tilde{\phi};\tilde{\psi}\},\{\phi;\psi\})$
is \emph{a dual framelet in $\Lp{2}$} if both $\{\tilde{\phi};\tilde{\psi}\}$ and $\{\phi;\psi\}$ are framelets in $\Lp{2}$ such that
\be \label{df}
\la f,g\ra=\sum_{k\in \Z} \la f, \tilde{\phi}(\cdot-k)\ra\la \phi(\cdot-k),g\ra+
\sum_{j=0}^\infty \sum_{k\in \Z}
\la f, \tilde{\psi}_{j;k}\ra \la \psi_{j;k},g\ra,\qquad \forall\, f,g\in \Lp{2}
\ee
with the above series converging absolutely.
For improved presentation and simplicity, in \eqref{df} we used the definition in \eqref{innerprod} for inner product of vector functions.
If in addition $\AS_0(\tilde{\phi};\tilde{\psi})$ and $\AS_0(\phi;\psi)$ are biorthogonal to each other, then a dual framelet $(\{\tilde{\phi};\tilde{\psi}\},\{\phi;\psi\})$ in $\Lp{2}$ is called \emph{a biorthogonal wavelet in $\Lp{2}$} or more precisely, a biorthogonal multiwavelet in $\Lp{2}$. It follows directly from \eqref{df} that every function $f\in \Lp{2}$ has the following framelet/wavelet expansion:
\be \label{df:repr}
f=\sum_{k\in \Z} \la f, \tilde{\phi}(\cdot-k)\ra \phi(\cdot-k)+
\sum_{j=0}^\infty \sum_{k\in \Z}
\la f, \tilde{\psi}_{j;k}\ra \psi_{j;k}
\ee
with the series converging unconditionally in $\Lp{2}$.
For $n\in \N\cup\{0\}$, the truncated expansions of \eqref{df:repr} are given by (see \cite{dhrs03})
\be \label{df:repr:trunc}
\mathcal{A}_n f:=\sum_{k\in \Z} \la f, \tilde{\phi}(\cdot-k)\ra \phi(\cdot-k)+
\sum_{j=0}^{n-1} \sum_{k\in \Z}
\la f, \tilde{\psi}_{j;k}\ra \psi_{j;k},\qquad f\in \Lp{2}.
\ee
Then $\mathcal{A}_n f\in \Lp{2}$ and $\lim_{n\to \infty} \|\mathcal{A}_n f-f\|_{\Lp{2}}=0$. We shall see in Proposition~\ref{prop:framelet} that
the truncated expansions in \eqref{df:repr:trunc} are closely linked to the well-known quasi-projection operators which we shall discuss here. By $\lLp{2}$ we denote the space of all locally square integrable functions, that is, $f\in \lLp{2}$ if $\int_{-N}^N |f(x)|^2 dx<\infty$ for all $N\in \N$. Let $\phi$ and $\tilde{\phi}$ be compactly supported vector functions in $(\Lp{2})^r$. For $n\in \N\cup\{0\}$, the \emph{quasi-projection operators} $\qp_n: \lLp{2}\rightarrow \lLp{2}$ associated with the compactly supported vector functions $\phi$ and $\tilde{\phi}$ are defined to be
\be \label{qp}
[\qp_n f](x)=\sum_{k\in \Z} \la f, \tilde{\phi}_{n;k}\ra \phi_{n;k}(x)=\sum_{k\in \Z} \la f, 2^n \tilde{\phi}(2^n\cdot-k)\ra \phi(2^n x-k),
\qquad n\in \N\cup\{0\}, f\in \lLp{2}.
\ee
Since $f\in \lLp{2}$ and $\tilde{\phi}$ is a compactly supported vector function in $(\Lp{2})^r$, the inner product $\la f, 2^n \tilde{\phi}(2^n\cdot-k)\ra$ is a well-defined row vector. On the other hand, since $\phi$ has compact support, for any given $x\in \R$, the summations in \eqref{qp} are actually finite. Therefore, $\qp_n f$ is well defined and we can easily observe that $\qp_n f\in \lLp{2}$. Hence, quasi-projection operators $\qp_n$ map $\lLp{2}$ to $\lLp{2}$.
Quasi-projection operators are well known in approximation theory, e.g., see \cite{hanbook,jz95,jia04,jj02}. In particular, we define $\qp:=\qp_0$, that is,
\be \label{qp:0}
\qp f:=\sum_{k\in \Z} \la f, \tilde{\phi}(\cdot-k)\ra \phi(\cdot-k),\qquad f\in \lLp{2}.
\ee
Then $[\qp_n f](2^{-n}\cdot)=\qp (f(2^{-n}\cdot))$ for all $n\in \N$.
We shall see in Proposition~\ref{prop:framelet} that
the truncated expansions $\mathcal{A}_n f$ in \eqref{df:repr:trunc} are linked to the quasi-projection operators through the identities $\mathcal{A}_n f=\qp_n f$ for all $f\in \Lp{2}$ and $n\in \N\cup\{0\}$.
Obviously, any polynomial $\pp$ belongs to $\lLp{2}$ and hence, $\qp_n \pp$ is well defined.
Recall that
the sign function is defined to be
\be \label{sgn}
\sgn(x)=1 \quad  \mbox{if}\;\; x>0; \qquad \sgn(0):=0; \qquad
\sgn(x):=-1\quad  \mbox{if}\;\; x<0.
\ee
Note that $\sgn\in \lLp{2}$. Hence, $\qp \sgn$ is well defined.
The Fourier transform in this paper is defined to be $\wh{f}(\xi):=\int_{\R} f(x) e^{-ix \xi}dx$ for $f\in \Lp{1}$ and is naturally extended to square integrable functions.

The definition of the Gibbs phenomenon for a general approximation scheme will be stated in Definition~\ref{def:gibbs}.
To study the Gibbs phenomenon of the truncated framelet expansions in \eqref{df:repr:trunc}, it is necessary to study the Gibbs phenomenon of the associated quasi-projection operators in \eqref{qp}. To do so, we need the following key identity in this paper on quasi-projection operators.

\begin{theorem}\label{thm:main}
Let $\phi$ and $\tilde{\phi}$ be $r\times 1$ vectors of compactly supported functions in $\Lp{2}$ such that
\be \label{basic:cond}
\ol{\wh{\tilde{\phi}}(0)}^\tp\wh{\phi}(0)=1
\quad \mbox{and}\quad
\ol{\wh{\tilde{\phi}}(0)}^\tp \wh{\phi}(2\pi k)=0 \qquad \forall\, k\in \Z\bs\{0\}.
\ee
Let $\qp$ be defined in \eqref{qp:0}.
Then $\sgn-\qp \sgn$ is a compactly supported function in $\Lp{2}$ and
\be \label{identity}
\la (\cdot), \sgn-\qp \sgn\ra=\frac{1}{6}
-\ol{\wh{\phi}(0)}^\tp (\kappa_1-\kappa_2)
-\ol{[\wh{\phi}]''(0)}^\tp \wh{\tilde{\phi}}(0)+i\ol{[\wh{\phi}]'(0)}^\tp \wh{\tilde{\phi}}(0)
-2i\ol{[\wh{\phi}]'(0)}^\tp \kappa_1,
\ee
where $(\cdot)$ stands for the identity function/polynomial such that $(x):=x$ for all $x\in \R$, and
\be \label{kappa}
\kappa_j:=\int_0^1 \sum_{n\in \Z} n^j \tilde{\phi}(x-n)dx, \qquad j\in \N\cup\{0\}.
\ee
\end{theorem}

We shall present the technical proof of Theorem~\ref{thm:main} in Section~\ref{sec:proof}.
It is well known (also see Section~2) that
\eqref{basic:cond} is equivalent to $\qp 1=1$. This condition is needed in order to guarantee that $\sgn-\qp \sgn$ is compactly supported.
%i.e., $\ol{\wh{\tilde{\phi}}(0)}^\tp \sum_{k\in \Z} \phi(\cdot-k)=1$.
As we shall see in this paper,
the identity \eqref{identity} in Theorem~\ref{thm:main} plays the key role in our study of the Gibbs phenomenon of framelet expansions and quasi-projection approximation.

The structure of the paper is as follows. Using Theorem~\ref{thm:main}, we shall define and discuss in Section~2 the Gibbs phenomenon for approximation by quasi-projection operators. In particular, we show that every approximation scheme by general quasi-projection operators with at least three approximation/accuracy orders exhibits the Gibbs phenomenon at the origin.
We further discuss in Section~2 about how to avoid the Gibbs phenomenon for approximation by quasi-projection operators without sacrificing the approximation orders and accuracy orders.
In Section~\ref{sec:wavelet} we address the Gibbs phenomenon of framelet expansions. In particular, we show that every dual framelet, if both its primal and dual framelet generators have at least two vanishing moments, must exhibit the Gibbs phenomenon at the origin.
We shall prove
Theorem~\ref{thm:main} in Section~\ref{sec:proof}.
Under the conditions $\qp1=1$ and all the entries of the vector function $\phi$ are continuous, in Section~\ref{sec:gibbs:x0} we shall study the Gibbs phenomenon at an arbitrary point for framelet expansions and approximation by quasi-projection operators. We show that every approximation scheme by general quasi-projection operators with at least three approximation/accuracy orders exhibits the Gibbs phenomenon at every point of the real line. Moreover, any dual framelet with at least two vanishing moments for its primal and dual framelet generators must exhibit the Gibbs phenomenon at all points.
Though most results in this paper can be generalized to quasi-projection operators in $\Lp{p}$ with $1\le p\le \infty$ and to wavelets and framelets not necessarily having compact support, for simplicity and for avoiding too much technicality, we only consider the space $\Lp{2}$ and framelets with compact support in this paper.

\section{Gibbs Phenomenon of Approximation by Quasi-projection Operators}
\label{sec:qp}

In this section we apply Theorem~\ref{thm:main} to study the Gibbs phenomenon of approximation by general quasi-projection operators.
Let us first give the definition of the Gibbs phenomenon of a general approximation scheme.

Generally, we often expand/represent functions in $\Lp{2}$ under various bases in $\Lp{2}$, e.g., orthogonal wavelet bases.
Therefore, we approximate a function $f\in \Lp{2}$
by a sequence $\{\qp_n f\}_{n\in \N}$ of functions in $\Lp{2}$ using some linear operators $\qp_n$ mapping $\Lp{2}$ to $\Lp{2}$ such that $\lim_{n\to \infty} \|\qp _n f-f\|_{\Lp{2}}=0$.
For piecewise smooth/analytic functions $f\in \Lp{2}$ with finitely many jump discontinuities, quite often we have $\lim_{n\to \infty} [\qp_n f](x)=f(x)$ uniformly in a neighborhood of every given point where $f$ is continuous (e.g., see Lemma~\ref{lem:qp:shift}). Therefore,  to study the Gibbs phenomenon, we only need to consider a special function with only one jump discontinuity. The definition of the Gibbs phenomenon under a general approximation scheme $\{\qp_n\}_{n\in \N}$ is defined as follows:

\begin{definition}\label{def:gibbs}
Let $x_0\in \R$ and $\eta$ be a compactly supported $C^\infty$ function on $\R$ such that $\eta(x)=1$ for all $x\in [-1,1]$. We say that a sequence $\{\qp_n\}_{n\in \N}$ of linear operators, mapping real-valued functions in $\Lp{2}$ into real-valued functions in $\Lp{2}$, exhibits the Gibbs phenomenon at the point $x_0$ if there exists a sequence $\{c_n\}_{n\in \N}$ of positive numbers such that $\lim_{n\to \infty} c_n=0$ and either
\be \label{r:gibbs}
\limsup_{n\to \infty} \mbox{ess-sup}_{x\in (x_0,x_0+c_n)} [\qp_n g](x)>1,
\ee
or
\be \label{l:gibbs}
\liminf_{n\to \infty} \mbox{ess-inf}_{x\in (x_0-c_n,x_0)} [\qp_n g](x)<-1,
\ee
where $g:=\eta(\cdot-x_0) \sgn(\cdot-x_0)$, which is smooth and continuous everywhere except at the point $x_0$.
\end{definition}

The inequality in
\eqref{r:gibbs} means that the approximation $\qp_n g$
overshoots $g$ by a fixed positive amount from the right-hand side of $x_0$.
%while \eqref{l:gibbs} means that the approximation $\qp_n g$ undershoots $g$ from the left side of $x_0$.
Note that in the above definition we do not require that all $\qp_n g$ should be continuous functions.
If all $\qp_n g$ are continuous (the most common case), then the above definition is equivalent to that there exists a sequence $\{c_n\}_{n\in \N}$ of positive numbers such that $\lim_{n\to \infty} c_n=0$ and either $\limsup_{n\to \infty} [\qp_n g](x_0+c_n)>1$ or
$\liminf_{n\to \infty} [\qp_n g](x_0-c_n)<-1$.

In this section we are particularly interested in the Gibbs phenomenon of approximation by quasi-projection operators.
Let us first discuss the Gibbs phenomenon at the origin (i.e., $x_0=0$) in Definition~\ref{def:gibbs}.
As discussed in \cite{kel96,shen11} for orthogonal wavelets, the Gibbs phenomenon at a point $x_0$ which is not a dyadic rational number is much more complicated.
For simplicity of presentation, we shall postpone our discussion on the Gibbs phenomenon at an arbitrary point to Section~\ref{sec:gibbs:x0}, the last section of this paper. In Section~\ref{sec:gibbs:x0}, we shall characterize and study the Gibbs phenomenon at an arbitrary point for approximation by quasi-projection operators and framelet expansions.
Since all involved functions in $\phi$ and $\tilde{\phi}$ in the quasi-projection operators
in \eqref{qp} have compact support and $\sgn\in \lLp{2}$, the functions $\qp_n \sgn$ are well defined and
$\qp_n \sgn$ agrees with $\qp_n (\eta \sgn)$ in a neighborhood of the origin for all large $n\in \N$. Therefore, we only need to study $\qp_n \sgn$ in a neighborhood of the origin, even though $\sgn \not \in \Lp{2}$.

\begin{lemma}\label{lem:qp:gibb}
Let $\phi$ and $\tilde{\phi}$ be $r\times 1$ vectors of compactly supported real-valued functions in $\Lp{2}$. Let
$\qp_n, n\in \N$ be the quasi-projection operators defined in \eqref{qp} and $\qp:=\qp_0$ in \eqref{qp:0}. Then
$\{\qp_n\}_{n\in \N}$ does not exhibit the Gibbs phenomenon at the origin if and only if
\be \label{qp:nogibbs}
[\qp \sgn](x)\le 1 \qquad a.e.\; x\in (0,\infty) \quad \mbox{and}\quad
[\qp \sgn](x)\ge -1\qquad a.e.\; x\in (-\infty,0).
\ee
\end{lemma}

\bp
%By the above discussion, we only need to consider the Gibbs phenomenon at the origin.
From the definition of the quasi-projection operators, we have the basic fact
$[\qp_n f](2^{-n}\cdot)=\qp (f(2^{-n}\cdot))$ for all $n\in \N$. In particular, by $\sgn(2^{-n}\cdot)=\sgn$, we have $[\qp_n \sgn](2^{-n} \cdot)=\qp \sgn$.

The sufficiency part is trivial, since \eqref{qp:nogibbs} implies $[\qp_n \sgn](x)=[\qp \sgn](2^n x)\le 1$ for almost every $x\in (0,\infty)$, and $[\qp_n \sgn](x)=[\qp \sgn](2^n x)\ge -1$ for almost every $x\in (-\infty,0)$.

Necessity. Suppose that $\{\qp_n \}_{n\in \N}$ does not exhibit the Gibbs phenomenon at the origin but
\eqref{qp:nogibbs} fails. Without loss of generality, we assume that there is a measurable set $E\subseteq (0,\infty)$ such that $E$ has a positive measure and $[\qp \sgn](x)>1$ for all $x\in E$. Therefore, there exists $c>0$ such that $C:=\mbox{ess-sup}_{x\in (0,c)} [\qp \sgn](x)>1$.
Define $c_n:=2^{-n} c>0$. Then $\lim_{n\to \infty} c_n=0$ and
\[
\mbox{ess-sup}_{x\in (0,c_n)}[\qp_n \sgn](x)=
\mbox{ess-sup}_{x\in (0,c)} [\qp \sgn](x)=C>1.
\]
Therefore, $\{\qp_n \}_{n\in \N}$ exhibits the Gibbs phenomenon at the origin, which is a contradiction to our assumption.
Consequently, \eqref{qp:nogibbs} must hold and we proved the necessity part.
\ep

In Section~\ref{sec:gibbs:x0},
under the conditions that $\qp 1=1$ and all the entries of $\phi$ are continuous,
we shall see that Lemma~\ref{lem:qp:gibb} is a special case of Theorem~\ref{thm:gibbsx0}, which characterizes the Gibbs phenomenon at an arbitrary point of approximation by quasi-projection operators.

Before addressing the Gibbs phenomenon of approximation by general quasi-projection operators, let us recall some well known results on quasi-projection operators from approximation theory.
For $m\in \N$, the Sobolev space $H^m(\R)$ consists of all functions $f\in \Lp{2}$ such that $f,f',\ldots, f^{(m)}\in \Lp{2}$ (all the derivatives are in the sense of distributions).
By $\PL_{m-1}$ we denote the set of all polynomials having degree less than $m$.
The following result is well known in approximation theory, e.g., see \cite{jz95,jia04}, \cite[Theorem~5.4.2]{hanbook} and references therein.

\begin{theorem}\label{thm:appr}
Let $\phi$ and $\tilde{\phi}$ be $r\times 1$ vectors of compactly supported functions in $\Lp{2}$. Let
$\qp_n, n\in \N\cup\{0\}$ be the quasi-projection operators defined in \eqref{qp}. For $m\in \N$,
$\{\qp_n\}_{n\in \N}$ has approximation order $m$, that is, there exists a positive constant $C$ such that
\be \label{apprord}
\|\qp_n f-f\|_{\Lp{2}}\le C 2^{-nm} \|f^{(m)}\|_{\Lp{2}},\qquad \forall\; f\in H^m(\R), n\in \N
\ee
if and only if $\qp \pp=\pp$ for all polynomials $\pp\in \PL_{m-1}$ of degree less than $m$, where $\qp:=\qp_0$ is defined in \eqref{qp:0}.
\end{theorem}

For two smooth functions $f$ and $g$, by $f(\xi)=g(\xi)+\bo(|\xi|^m)$ as $\xi\to 0$ we mean $f^{(j)}(0)=g^{(j)}(0)$ for all $j=0,\ldots,m-1$.
By $\td$ we denote the Dirac sequence such that $\td(0)=1$ and $\td(k)=0$ for all $k\in \Z\bs\{0\}$.
We say that $\{\qp_n\}_{n\in \N}$ has \emph{accuracy order $m$} if $\qp \pp=\pp$ for all $\pp\in \PL_{m-1}$, which is also equivalent to that $\qp_n \pp =\pp$ for all $\pp\in \PL_{m-1}$ and $n\in \N\cup\{0\}$.
It is also well known (e.g., see \cite[Proposition~5.5.2]{hanbook}) that $\qp \pp=\pp$ for all $\pp\in \PL_{m-1}$ if and only if
\be \label{polyprod}
\ol{\wh{\tilde{\phi}}(\xi)}^\tp \wh{\phi}(\xi+2\pi k)=\td(k)+\bo(|\xi|^m),\qquad \xi\to 0\quad \mbox{for all}\quad k\in \Z.
\ee
Similarly, we can define quasi-projection operators $\tilde{\qp}_n$ by switching the roles of $\phi$ and $\tilde{\phi}$ as follows:
\be \label{qp:dual}
\tilde{\qp}:=\tilde{\qp}_0\quad \mbox{and}\quad
\tilde{\qp}_n f:=\sum_{k\in \Z} \la f, 2^n \phi(2^n \cdot-k) \ra \tilde{\phi}(2^n \cdot-k),\qquad  n\in \N\cup\{0\}, f\in \lLp{2}.
\ee
Then $\tilde{\qp} \pp=\pp$ for all $\pp\in \PL_{m-1}$ if and only if
\be \label{tildepolyprod}
\ol{\wh{\phi}(\xi)}^\tp \wh{\tilde{\phi}}(\xi+2\pi k)=\td(k)+\bo(|\xi|^m),\qquad \xi\to 0\quad \mbox{for all}\quad k\in \Z.
\ee

A simple sufficient condition
%for $\phi$ and $\tilde{\phi}$ to satisfy \eqref{qp:nogibbs}
for $\{\qp_n\}_{n\in \N}$  to be free of the Gibbs phenomenon is as follows.

\begin{prop}\label{prop:nogibb}
Let $\phi$ and $\tilde{\phi}$ be $r\times 1$ vectors of compactly supported real-valued functions in $\Lp{2}$ such that \eqref{basic:cond} is satisfied, i.e., $\qp 1=1$, where $\qp$ is defined in \eqref{qp:0}.
\begin{enumerate}
\item[(i)] If all the entries in the vector functions $\phi$, $\int_{-\infty}^{k} \tilde{\phi}(x) dx$ and $\int_k^\infty \tilde{\phi}(x) dx$ are nonnegative for all $k\in \Z$, then $-1\le [\qp \sgn](x)\le 1$ for almost every $x\in \R$ and in particular, \eqref{qp:nogibbs} holds.
\item[(ii)] If all the entries in $\phi$ and $\tilde{\phi}$ are nonnegative, then $-1\le [\qp \sgn](x)\le 1$ for almost every $x\in \R$ and \eqref{qp:nogibbs} holds, but $\{\qp_n\}_{n\in \N}$ has accuracy order no more than two, where $\qp_n, n\in \N$ are defined in \eqref{qp}
\end{enumerate}
\end{prop}

\begin{comment}
By \eqref{basic:cond} and $\la 1, \tilde{\phi}(\cdot-k)\ra=\ol{\wh{\tilde{\phi}}(0)}^\tp$, we have
%
\be \label{basic:cond:2}
\ol{\wh{\tilde{\phi}}(0)}^\tp \sum_{k\in \Z} \phi(\cdot-k)=\sum_{k\in \Z} \la 1, \tilde{\phi}(\cdot-k)\ra \phi(\cdot-k)=1.
\ee
%
Noting that $\tilde{\phi}$ is real-valued and $\la \sgn, \tilde{\phi}(\cdot-k)\ra^\tp =\wh{\tilde{\phi}}(0)-2\int_{-\infty}^{-k} \tilde{\phi}(t) dt$ for all $k\in \Z$, we have
\[
[\qp \sgn](x)-1=-2\sum_{k\in \Z}
\left(\int_{-\infty}^{-k} \tilde{\phi}(t) dt\right)^\tp \phi(x-k)\le 0,
\]
for almost every $x\in \R$, since by item (i) all entries in both $\phi$ and $\int_{-\infty}^{-k} \tilde{\phi}(t) dt$ are nonnegative.
Hence, we proved $[\qp \sgn](x)\le 1$ for almost every $x\in \R$.
On the other hand,
\[
[\qp \sgn](x)+1=\sum_{k\in \Z}
2\left(\wh{\tilde{\phi}}(0)-\int_{-\infty}^{-k} \tilde{\phi}(t) dt\right)^\tp \phi(x-k)
=2\sum_{k\in \Z}
\left(\int_{-k}^{\infty} \tilde{\phi}(t) dt\right)^\tp \phi(x-k)\ge 0.
\]
Hence, $(\qp \sgn)(x)\ge -1$ for almost every $x\in \R$.
This proves item (i).
\end{comment}

\bp Note that $\la \tilde{\phi}(\cdot-k),1\ra=\wh{\tilde{\phi}}(0)$ and $\la \tilde{\phi}(\cdot-k),\sgn\ra=2\int_0^\infty \tilde{\phi}(y-k) dy-\wh{\tilde{\phi}}(0)$. Hence,
\be \label{Qsgn}
[\qp \sgn](x)=2\int_0^\infty K(x,y) dy-\ol{\wh{\tilde{\phi}}(0)}^\tp \sum_{k\in \Z} \phi(x-k) \quad \mbox{with}\quad
K(x,y):=\sum_{k\in \Z} \ol{\tilde{\phi}(y-k)}^\tp \phi(x-k).
\ee
Since \eqref{basic:cond} holds (i.e., $\qp 1=1$) and $\tilde{\phi}$ is real-valued,
by the assumptions in item (i), we have
\[
[\qp \sgn](x)=2\int_0^\infty K(x,y) dy-1=2\sum_{k\in \Z}\left(\int_{-k}^\infty \tilde{\phi}(y) dy\right)^\tp \phi(x-k)-1 \ge -1
\]
and
\[
[\qp \sgn](x)=1-2\int_{-\infty}^0 K(x,y) dy=
1-2\sum_{k\in \Z}\left(\int^{-k}_{-\infty} \tilde{\phi}(y) dy\right)^\tp \phi(x-k)\le 1
\]
for almost all $x\in \R$. Hence, $-1\le [\qp \sgn](x)\le 1$ for almost every $x\in \R$. This proves item (i).

Obviously, the conditions in item (ii) imply the conditions in item (i). Therefore, \eqref{qp:nogibbs} must hold.
We use proof by contradiction to show that $\{\qp_n\}_{n\in \N}$ has accuracy order no more than two. Suppose that
$\{\qp_n\}_{n\in \N}$ has accuracy order at least three.
Define $\eta(x):=\int_{\R} \ol{\tilde{\phi}(x+y)}^\tp \phi(y) dy$. By \eqref{polyprod} with $m=3$,
we must have $\wh{\eta}(\xi)=\ol{\wh{\tilde{\phi}}(\xi)}^\tp \wh{\phi}(\xi)=1+\bo(|\xi|^3)$ as $\xi \to 0$, from which we conclude that $[\wh{\eta}]''(0)=0$ and $\wh{\eta}(0)=1$. By $0=[\wh{\eta}]''(0)=\int_{\R} \eta(x)(-i x)^2 dx$, we must have $\int_\R x^2 \eta(x) dx=0$. However, since both $\tilde{\phi}$ and $\phi$ are nonnegative, the scalar function $\eta$ must be nonnegative. The condition $\int_\R x^2 \eta(x) dx=0$ will then force $\eta=0$, a contradiction to $\wh{\eta}(0)=1$.
This proves that
$\{\qp_n\}_{n\in \N}$ must have accuracy order no more than two. This proves item (ii).
\ep

Under the condition in \eqref{basic:cond} (i.e., $\qp 1=1$), by \eqref{Qsgn}, the criterion in \eqref{qp:nogibbs} of Lemma~\ref{lem:qp:gibb} for $\{\qp_n\}_{n\in \N}$  to be free of the Gibbs phenomenon is equivalent to
\be \label{gibbs:K}
\int_0^\infty K(x,y) dy\le 1 \quad a.e.\, x>0
\quad \mbox{and}\quad
\int_0^\infty K(x,y) dy\ge 0 \quad a.e.\, x<0.
\ee
The criterion in \eqref{gibbs:K} for $\{\qp_n\}_{n\in \N}$  to be free of the Gibbs phenomenon is already known in \cite[Theorem~3.1]{kel96} for orthogonal wavelets and in \cite{gc95} for more general expansions.
Many functions in wavelet analysis and approximation theory are nonnegative. One important family of such functions are B-splines. \emph{The B-spline function $B_m$ of order $m$} is defined by
\be \label{bsplines}
B_1:=\chi_{(0,1]} \quad \mbox{and}\quad
B_m:=B_{m-1}*B_1=\int_0^1 B_{m-1}(\cdot-t) dt,\qquad m\in \N.
\ee
It is easy to check that $B_m$ is a nonnegative piecewise polynomial with support $[0,m]$ and $\wh{B_m}(\xi)=(\frac{1-e^{-i\xi}}{i\xi})^m$. Therefore, $\wh{B_m}(0)=1$ and $\wh{B_m}(\xi+2\pi k)=\bo(|\xi|^m)$ as $\xi \to 0$ for all $k\in \Z \bs\{0\}$.

For the Gibbs phenomenon of approximation by general quasi-projection operators, we have

\begin{theorem}\label{thm:qp}
Let $\phi$ and $\tilde{\phi}$ be $r\times 1$ vectors of compactly supported functions in $\Lp{2}$. Let
$\qp_n$ be the quasi-projection operators defined in \eqref{qp} and $\qp:=\qp_0$ in \eqref{qp:0}. If \eqref{basic:cond} holds (i.e., $\qp 1=1$) and
\be \label{basic:cond:dual}
\ol{\wh{\phi}(\xi)}^\tp \wh{\tilde{\phi}}(\xi+2\pi k)=\td(k)+\bo(|\xi|^2),\qquad \xi \to 0\quad \mbox{for all}\;\; k\in \Z,
\ee
then
\be \label{sgn:qp}
\la (\cdot), \sgn-\qp \sgn\ra=-[\ol{\wh{\phi}}^\tp \wh{\tilde{\phi}}]''(0),
\ee
where $(\cdot)$ stands for the linear polynomial with $(x):=x$ for all $x\in \R$.
If in addition $[\ol{\wh{\phi}}^\tp \wh{\tilde{\phi}}]''(0)=0$, both $\phi$ and $\tilde{\phi}$ are real-valued, and
\be \label{sgn:cond}
\qp \sgn\ne \sgn \quad \mbox{on some set of positive measure},
\ee
then $\{\qp_n \}_{n\in \N}$ must exhibit the Gibbs phenomenon at the origin.
\end{theorem}

\bp
Since \eqref{basic:cond} holds, by Theorem~\ref{thm:main}, the identity in \eqref{identity} must be true.
By the definition of the quasi-projection operators $\tilde{\qp}$ and $\tilde{\qp}_n$ in \eqref{qp:dual},
since \eqref{basic:cond:dual} is just \eqref{tildepolyprod} with $m=2$, we conclude that \eqref{basic:cond:dual} is equivalent to saying that $\{\tilde{\qp}_n\}_{n\in \N}$ has accuracy order two, i.e., $\tilde{\qp}1=1$ and $\tilde{\qp} x=x$. By calculation, we have
$\la (\cdot), \phi(\cdot-k)\ra=\la (\cdot)+k, \phi\ra=-i\ol{[\wh{\phi}]'(0)}^\tp+k \ol{\wh{\phi}(0)}^\tp$. Hence,
\be \label{qp:dual:x}
x=\tilde{\qp} x=\sum_{k\in \Z} \la (\cdot), \phi(\cdot-k)\ra \tilde{\phi}(x-k)=\sum_{k\in \Z} \left( -i\ol{[\wh{\phi}]'(0)}^\tp+k \ol{\wh{\phi}(0)}^\tp
\right) \tilde{\phi}(x-k).
\ee
Multiplying $x$ to both sides of the above identity in \eqref{qp:dual:x} and integrating on $[0,1]$, by $x=(x-k)+k$, we have
\begin{align*}
\frac{1}{3}=\int_0^1 x^2 dx
&=-i\ol{[\wh{\phi}]'(0)}^\tp\sum_{k\in \Z} \int_0^1 x\tilde{\phi}(x-k)dx+\ol{\wh{\phi}(0)}^\tp \sum_{k\in \Z} \int_0^1 kx \tilde{\phi}(x-k)dx\\
&=-i\ol{[\wh{\phi}]'(0)}^\tp\left(\int_0^1 \sum_{k\in \Z} (x-k)\tilde{\phi}(x-k)dx+\int_0^1 \sum_{k\in \Z} k\tilde{\phi}(x-k) dx\right)\\
&\qquad +
\ol{\wh{\phi}(0)}^\tp \left(\int_0^1 \sum_{k\in \Z} k(x-k) \tilde{\phi}(x-k)dx+\int_0^1 \sum_{k\in \Z} k^2 \tilde{\phi}(x-k) dx\right)
\\
&=-i\ol{[\wh{\phi}]'(0)}^\tp\left(\int_\R  x\tilde{\phi}(x)dx+\kappa_1\right)+
\ol{\wh{\phi}(0)}^\tp (\kappa_*+\kappa_2)\\
&=\ol{[\wh{\phi}]'(0)}^\tp [\wh{\tilde{\phi}}]'(0)
-i\ol{[\wh{\phi}]'(0)}^\tp \kappa_1+\ol{\wh{\phi}(0)}^\tp
\kappa_*+\ol{\wh{\phi}(0)}^\tp \kappa_2,
\end{align*}
where $\kappa_1,\kappa_2$ are defined in \eqref{kappa}, and
$\kappa_*:=\int_{0}^{1} \sum_{k\in \Z} k(x-k) \tilde{\phi}(x-k)dx$.
Since $\la 1, \phi(\cdot-k)\ra=\ol{\wh{\phi}(0)}^\tp$ and
$\tilde{\qp} 1=1$, we have
\[
1=[\tilde{\qp} 1](x)=\sum_{k\in \Z}\la 1, \phi(\cdot-k)\ra \tilde{\phi}(x-k)=
\ol{\wh{\phi}(0)}^\tp \sum_{k\in \Z} \tilde{\phi}(x-k).
\]
Multiplying $x^2$ on both sides of the above identity and integrating on $[0,1]$, we have
\begin{align*}
\frac{1}{3}&=\int_0^1 x^2 dx
=\ol{\wh{\phi}(0)}^\tp \sum_{k\in \Z} \int_0^1 x^2 \tilde{\phi}(x-k)dx
=\ol{\wh{\phi}(0)}^\tp \int_0^1 \sum_{k\in \Z} \Big((x-k)^2+2k (x-k)+k^2\Big) \tilde{\phi}(x-k)dx\\
&=\ol{\wh{\phi}(0)}^\tp \left(
\int_0^1 \sum_{k\in \Z} (x-k)^2 \tilde{\phi}(x-k)dx+
\int_0^1 \sum_{k\in \Z} 2k(x-k) \tilde{\phi}(x-k)dx+
\int_0^1 \sum_{k\in \Z} k^2 \tilde{\phi}(x-k)dx\right)\\
&=\ol{\wh{\phi}(0)}^\tp \left( \int_\R x^2 \tilde{\phi}(x) dx+2\kappa_*+\kappa_2\right)\\
&=-\ol{\wh{\phi}(0)}^\tp [\wh{\tilde{\phi}}]''(0)
+2\ol{\wh{\phi}(0)}^\tp \kappa_*+
\ol{\wh{\phi}(0)}^\tp  \kappa_2.
\end{align*}
That is, we proved that \eqref{basic:cond:dual} implies
\[
\frac{1}{3}=
\ol{[\wh{\phi}]'(0)}^\tp [\wh{\tilde{\phi}}]'(0)
-i\ol{[\wh{\phi}]'(0)}^\tp \kappa_1+\ol{\wh{\phi}(0)}^\tp
\kappa_*+\ol{\wh{\phi}(0)}^\tp \kappa_2,
\qquad \frac{1}{3}=
-\ol{\wh{\phi}(0)}^\tp [\wh{\tilde{\phi}}]''(0)
+2\ol{\wh{\phi}(0)}^\tp \kappa_*+
\ol{\wh{\phi}(0)}^\tp  \kappa_2.
\]
Multiplying the first identity by $2$ and subtracting the second identity, we end up with
\be \label{eqkappa*}
\frac{1}{3}=2\ol{[\wh{\phi}]'(0)}^\tp [\wh{\tilde{\phi}}]'(0)-2i \ol{[\wh{\phi}]'(0)}^\tp \kappa_1 +\ol{\wh{\phi}(0)}^\tp \kappa_2+\ol{\wh{\phi}(0)}^\tp [\wh{\tilde{\phi}}]''(0).
\ee
Integrating over $[0,1]$ on both sides of \eqref{qp:dual:x}, we have
\[
\frac{1}{2}=\int_0^1 x dx=
-i\ol{[\wh{\phi}]'(0)}^\tp \sum_{k\in \Z} \int_0^1 \tilde{\phi}(x-k) dx+\ol{\wh{\phi}(0)}^\tp
\int_0^1 \sum_{k\in \Z} k \tilde{\phi}(x-k) dx=
-i\ol{[\wh{\phi}]'(0)}^\tp \wh{\tilde{\phi}}(0)+\ol{\wh{\phi}(0)}^\tp \kappa_1,
\]
from which we have $\ol{\wh{\phi}(0)}^\tp \kappa_1=\frac{1}{2}+i\ol{[\wh{\phi}]'(0)}^\tp \wh{\tilde{\phi}}(0)$.
Now we conclude from \eqref{identity} in Theorem~\ref{thm:main} that
\begin{align*}
\la (\cdot), \sgn-\qp \sgn\ra
&=\frac{1}{6}
-\ol{\wh{\phi}(0)}^\tp (\kappa_1-\kappa_2)
-\ol{[\wh{\phi}]''(0)}^\tp \wh{\tilde{\phi}}(0)+i\ol{[\wh{\phi}]'(0)}^\tp \wh{\tilde{\phi}}(0)
-2i\ol{[\wh{\phi}]'(0)}^\tp \kappa_1\\
&=
\frac{1}{6}
-\left(\frac{1}{2}+i\ol{[\wh{\phi}]'(0)}^\tp \wh{\tilde{\phi}}(0)\right)
+\ol{\wh{\phi}(0)}^\tp \kappa_2
-\ol{[\wh{\phi}]''(0)}^\tp \wh{\tilde{\phi}}(0)+i\ol{[\wh{\phi}]'(0)}^\tp \wh{\tilde{\phi}}(0)
-2i\ol{[\wh{\phi}]'(0)}^\tp \kappa_1\\
&=-\frac{1}{3}+\ol{\wh{\phi}(0)}^\tp \kappa_2
-\ol{[\wh{\phi}]''(0)}^\tp \wh{\tilde{\phi}}(0)
-2i\ol{[\wh{\phi}]'(0)}^\tp \kappa_1\\
&=
-2\ol{[\wh{\phi}]'(0)}^\tp [\wh{\tilde{\phi}}]'(0)-\ol{\wh{\phi}(0)}^\tp [\wh{\tilde{\phi}}]''(0)
-\ol{[\wh{\phi}]''(0)}^\tp \wh{\tilde{\phi}}(0)
=-[ \ol{\wh{\phi}}^\tp \wh{\tilde{\phi}}]''(0),
\end{align*}
where we used \eqref{eqkappa*} in the second-to-last identity.
Hence, \eqref{sgn:qp} must be true.

If in addition $[\ol{\wh{\phi}}^\tp \wh{\tilde{\phi}}]''(0)=0$, then it follows directly from \eqref{sgn:qp} that $\la (\cdot), \sgn-\qp \sgn\ra=0$.
Suppose that $\{\qp_n\}_{n\in \N}$ does not exhibit the Gibbs phenomenon at the origin. By Lemma~\ref{lem:qp:gibb}, \eqref{qp:nogibbs} must hold. Consequently,
$x(\sgn-\qp \sgn)(x)\ge 0$ for almost every $x\in \R$.
By $\la (\cdot), \sgn-\qp \sgn \ra=0$, we must have $x(\sgn-\qp \sgn)(x)=0$ for almost every $x\in \R$, which is a contradiction to our assumption in \eqref{sgn:cond}. Therefore, $\{\qp_n\}_{n\in \N}$ must exhibit the Gibbs phenomenon at the origin.
\ep

It is easy to see that \eqref{sgn:cond} is often true. For example, if $\phi$ is continuous, then \eqref{sgn:cond} must hold.
As a direct consequence of Theorem~\ref{thm:qp}, the following result says that quasi-projection operators having accuracy order higher than two must exhibit the Gibbs phenomenon.

\begin{cor}\label{cor:qp}
Let $\phi$ be an $r\times 1$ vector of compactly supported functions in $\Lp{2}$. Recall that the associated quasi-projection operators are defined by
\be \label{qp:phi}
\qp_n f=\sum_{k\in \Z} \la f, 2^n \phi(2^n \cdot-k)\ra \phi(2^n\cdot-k), \qquad  n\in \N\cup\{0\}, f\in \lLp{2}.
\ee
If $\{\qp_n\}_{n\in \N}$ has accuracy order higher than two, then $\la (\cdot), \sgn-\qp \sgn\ra=0$, where $\qp:=\qp_0$. If in addition \eqref{sgn:cond} holds, then $\{\qp_n\}_{n\in \N}$ must exhibit the Gibbs phenomenon at the origin.
\end{cor}

\bp Since $\{\qp_n\}_{n\in \N}$ has accuracy order higher than two, we must have $\qp 1=1, \qp x=x$ and $\qp x^2=x^2$; or equivalently,
\be \label{apprord:2}
\ol{\wh{\phi}(\xi)}^\tp \wh{\phi}(\xi+2\pi k)=\td(k)+\bo(|\xi|^3),\qquad \xi\to 0, k\in \Z.
\ee
Hence, both \eqref{basic:cond} and \eqref{basic:cond:dual} are satisfied with $\tilde{\phi}:=\phi$. Moreover, \eqref{apprord:2} further implies $\ol{\wh{\phi}(\xi)}^\tp \wh{\phi}(\xi)=1+\bo(|\xi|^3)$ as $\xi \to 0$, from which we have $[\ol{\wh{\phi}}^\tp \wh{\phi}]''(0)=0$.
Now the conclusion follows directly from Theorem~\ref{thm:qp}.
\ep

In the rest of this section we discuss how to avoid the Gibbs phenomenon while preserving higher accuracy order.
Let $\phi$ and $\tilde{\phi}$ be $r\times 1$ vectors of compactly supported functions in $\Lp{2}$. By $\si(\phi)$ we denote the space of all functions $\sum_{k\in \Z} v(k) \phi(\cdot-k)$ for all sequences $v: \Z\rightarrow \C^{1\times r}$. If the quasi-projection operators $\qp_n$ in \eqref{qp} have accuracy order $m$, then $\qp \pp=\pp$ for all $\pp\in \PL_{m-1}$, where $\qp:=\qp_0$ is defined in \eqref{qp:0}. In particular, we have $\PL_{m-1}\subset \si(\phi)$.
If $\PL_{m-1}\subset \si(\phi)$, then it is interesting to ask whether there exists a compactly supported vector function $\tilde{\phi}$ such that $\{\qp_n \}_{n\in \N}$ in \eqref{qp} has accuracy order $m$ and is free of the Gibbs phenomenon.
The following result partially answers this question.

\begin{prop}\label{prop:nogibbs}
Let $\phi$ be a nonnegative compactly supported real-valued function in $\Lp{2}$ such that $\wh{\phi}(0)=1$ and $\wh{\phi}(\xi+2\pi k)=\bo(|\xi|^m)$ as $\xi \to 0$ for all $k\in \Z\bs\{0\}$. Then there exists a compactly supported real-valued function $\tilde{\phi}\in \Lp{2}$ such that \eqref{polyprod} is satisfied,
$\sum_{k\in \Z} \tilde{\phi}(\cdot-k)=1$, and $\{\qp_n\}_{n\in \N}$ has accuracy order $m$ and is free of the Gibbs phenomenon at the origin, where $\qp_n, n\in \N$ are defined in \eqref{qp}.
\end{prop}

\bp
Define $d_j:=i^j [1/\ol{\wh{\phi}}]^{(j)}(0)$ for $j=0,\ldots,m-1$. Since $\phi$ is real-valued, we can check that all $d_j$ are real numbers. Moreover, $d_0=1$ and $d_1=\int_\R x\phi(x) dx$. Let $N$ be the unique integer such that $0\le d_1-N+\frac{1}{2}<1$.
Consider $N=x_0<x_1<\cdots<x_{m-1}=N+1$ and $c_0,\ldots,c_{m-1}\in \R$. Define
\[
\eta:=\sum_{k=1}^{m-1} c_k \chi_{[x_{k-1},x_k]} \quad \mbox{and}\quad \tilde{\phi}:=\eta-\eta(1+\cdot)+\chi_{(N-1,N]}.
\]
Then it is trivial that $\eta$ is supported inside $[N,N+1]$, $\tilde{\phi}$ is supported inside $[N-1,N+1]$, and
$\sum_{k\in \Z} \tilde{\phi}(\cdot-k)=1$.
We now choose the unknowns $\{c_k\}_{k=1}^{m-1}$ so that
\be \label{moment:match}
\wh{\tilde{\phi}}(\xi)=\frac{1}{\ol{\wh{\phi}(\xi)}}+\bo(|\xi|^m), \qquad \xi \to 0.
\ee
By the definition of $\tilde{\phi}$ and $d_j$, we see that \eqref{moment:match} is equivalent to
\[
\int_\R x^j \tilde{\phi}(x)dx=
\int_\R x^j \Big(\eta(x)-\eta(x+1)+\chi_{(N-1,N]}(x)\Big)dx=d_j,\qquad j=0,\ldots,m-1.
\]
Note that $\int_\R x^j \chi_{(N-1,N]}(x) dx=\frac{N^{j+1}-(N-1)^{j+1}}{j+1}$,
$\int_\R x^j \eta(x+1) dx=\int_\R (x-1)^j \eta(x) dx$, and
\[
\int_{\R} x^j (\eta(x)-\eta(x+1))dx=\int_\R (x^j-(x-1)^j) \eta(x) dx=\sum_{k=1}^{m-1} c_k \int_{x_{k-1}}^{x_k} (x^j-(x-1)^j) dx.
\]
Consequently, \eqref{moment:match} is further equivalent to that for all $j=0,\ldots,m-1$,
\be \label{eqn:c}
\sum_{k=1}^{m-1} c_k \int_{x_{k-1}}^{x_k} (x^j-(x-1)^j) dx
=d_j-\frac{N^{j+1}-(N-1)^{j+1}}{j+1}.
\ee
For $j=0$, by $d_0=1$, both sides of the above equation in \eqref{eqn:c} are zero.
We now prove that the system of linear equations in \eqref{eqn:c} for $j=1,\ldots,m-1$ has a unique solution $\{c_k\}_{k=1}^{m-1}$. In fact, we consider its dual problem:
If
\be \label{eqn:c:dual}
\int_{x_{k-1}}^{x_k} \pp(x) dx=0 \quad \mbox{for all}\;\; k=1,\ldots, m-1 \quad \mbox{with}\quad \pp(x):=\sum_{j=1}^{m-1} b_j (x^j-(x-1)^j),
\ee
we must prove that $b_1=\cdots=b_{m-1}=0$. Let $\pq$ be the unique polynomial such that $\pq'=\pp$ and $\pq(x_0)=0$. Then $\deg(\pq)\le m-1$ by $\deg(\pp)\le m-2$.
Since $\int_{x_{k-1}}^{x_k} \pp(x) dx=\pq(x_k)-\pq(x_{k-1})$ and $\pq(x_0)=0$, we see that \eqref{eqn:c:dual} is equivalent to
$\pq(x_1)=\cdots=\pq(x_{m-1})=0$. Since $\pq(x_0)=0$, the polynomial $\pq$ has $m$ distinct roots $x_0,\ldots, x_{m-1}$. Then $\deg(\pq)\le m-1$ will force $\pq=0$ and consequently, $\pp=\pq'=0$, from which we must have $b_1=\cdots=b_{m-1}=0$.
Since all $d_j$ are real numbers,
this proves that the system of linear equations in \eqref{eqn:c} has a unique real-valued solution $\{c_k\}_{k=1}^{m-1}$.
Therefore, \eqref{polyprod} is satisfied and
$\{\qp_n\}_{n\in \N}$ has accuracy order $m$.

Considering $j=1$ in \eqref{eqn:c}, we have
\[
\int_\R \eta(x) dx=\sum_{k=1}^{m-1} c_k (x_k-x_{k-1})=d_1-\frac{N^2-(N-1)^2}{2}=d_1-N+\frac{1}{2}.
\]
Note that $\tilde{\phi}=\eta$ on $[N,N+1]$. By our choice of $N\in \Z$, we conclude that the number $\int_{N}^{N+1} \tilde{\phi}(x) dx=\int_\R \eta(x)dx=d_1-N+\frac{1}{2}$ must lie on the interval $[0,1]$.
On the other hand, since $\tilde{\phi}=\chi_{(N-1,N]}-\eta(1+\cdot)$ on $[N-1,N]$, we have
\[
\int_{N-1}^N \tilde{\phi}(x) dx=1-\int_\R \eta(1+x) dx=
1-\left(d_1-N+\frac{1}{2}\right)
\]
must lie on $[0,1]$ by $0\le d_1-N+\frac{1}{2}<1$. Since $\tilde{\phi}$ is supported inside $[N-1,N+1]$, now it follows from
item (i) of Proposition~\ref{prop:nogibb} that $\{\qp_n\}_{n\in \N}$ is free of the Gibbs phenomenon at the origin.
\ep

According to Theorem~\ref{thm:qp}, for $m\ge 3$, Proposition~\ref{prop:nogibbs} is optimal in the sense that any desired compactly supported function $\tilde{\phi}$ in Proposition~\ref{prop:nogibbs} cannot satisfy the additional condition: $[\wh{\tilde{\phi}}]'(2\pi k)=0$ for all $k\in \Z\bs\{0\}$. Otherwise, Theorem~\ref{thm:qp} tells us that $\{\qp_n\}_{n\in \N}$ must exhibit the Gibbs phenomenon at the origin, a contradiction to the claim in Proposition~\ref{prop:nogibbs}.

\section{Gibbs Phenomenon of Wavelet and Framelet Expansions}
\label{sec:wavelet}

In this section we study the Gibbs phenomenon of wavelet and framelet expansions. As we shall see in this section, wavelets and framelets having high vanishing moments must exhibit the Gibbs phenomenon at the origin.

One of the main features of wavelets and framelets $\{\phi;\psi\}$ is sparse representation, which is largely due to the vanishing moments of $\psi$  (see \cite{daubook}). We say that $\psi$ has \emph{$m$ vanishing moments}  if $\int_\R x^j \psi(x)dx=0$ for all $j=0,\ldots,m-1$, or equivalently, $\wh{\psi}^{(j)}(0)=0$ for all $j=0,\ldots,m-1$.
In particular, we define $\vmo(\psi):=m$ with $m$ being the largest such nonnegative integer.

Before discussing the Gibbs phenomenon of wavelet and framelet expansions, we have the following result, which is essentially known in the literature but we shall provide a proof here for the convenience of the reader.

\begin{prop}\label{prop:framelet}
Let $\phi, \tilde{\phi}$ be $r\times 1$ vectors of compactly supported functions in $\Lp{2}$ and $\psi, \tilde{\psi}$ be $s\times 1$ vectors of compactly supported functions in $\Lp{2}$.
Suppose that $(\{\tilde{\phi};\tilde{\psi}\},\{\phi;\psi\})$ is a dual framelet in $\Lp{2}$. Then
\begin{enumerate}
\item[(i)] $\mathcal{A}_nf=\qp_n f$ for all $f\in \Lp{2}$ and $n\in \N\cup\{0\}$, where $\mathcal{A}_n$ are the operators in \eqref{df:repr:trunc} for the truncated framelet expansions and $\qp_n$ are the quasi-projection operators defined in \eqref{qp}.

\item[(ii)] $\qp \pp= \pp$ for all $\pp\in \PL_{m-1}$ with $m:=\vmo(\tilde{\psi})$, where $\qp \pp:=\sum_{k\in \Z} \la \pp, \tilde{\phi}(\cdot-k)\ra \phi(\cdot-k)$.
\end{enumerate}
\end{prop}

\bp Item (i) is essentially known in \cite{han10,han12}. Indeed, for $J\in \Z$, replacing $f$ and $g$ in \eqref{df} by $2^{-J/2}f(2^{-J}\cdot)$ and $2^{-J/2}g(2^{-J}\cdot)$, respectively, we have
\begin{align*}
\la f,g\ra&=
\la 2^{-J/2}f(2^{-J}\cdot),2^{-J/2}g(2^{-J}\cdot)\ra=
\sum_{k\in \Z} \la 2^{-J/2}f(2^{-J}\cdot), \tilde{\phi}(\cdot-k)\ra\la \phi(\cdot-k),2^{-J/2}g(2^{-J}\cdot)\ra\\
&\qquad\qquad\qquad +
\sum_{j=0}^\infty \sum_{k\in \Z}
\la 2^{-J/2}f(2^{-J}\cdot), \tilde{\psi}_{j;k}\ra \la \psi_{j;k},2^{-J/2}g(2^{-J}\cdot)\ra\\
&=\sum_{k\in \Z} \la f, \tilde{\phi}_{J;k}\ra\la \phi_{J;k},g\ra+
\sum_{j=J}^\infty \sum_{k\in \Z}
\la f, \tilde{\psi}_{j;k}\ra \la \psi_{j;k},g\ra.
\end{align*}
Considering the differences between the levels $J=n-1$ and $J=n$ of the above identities, we conclude that
\be \label{cascade}
\sum_{k\in \Z} \la f, \tilde{\phi}_{n-1;k}\ra\la \phi_{n-1;k},g\ra+
\sum_{k\in \Z}
\la f, \tilde{\psi}_{n-1;k}\ra \la \psi_{n-1;k},g\ra
=\sum_{k\in \Z} \la f, \tilde{\phi}_{n;k}\ra\la \phi_{n;k},g\ra,\qquad n\in \Z.
\ee
Consequently, it follows directly from the above identities that
\[
\la \mathcal{A}_n f,g\ra=\sum_{k\in \Z} \la f, \tilde{\phi}_{n;k}\ra\la \phi_{n;k},g\ra=
\la \qp_n f, g\ra
\]
for all $f,g\in \Lp{2}$. Therefore, we must have $\mathcal{A}_n f=\qp_n f$ for all $f\in \Lp{2}$ and $n\in \N\cup\{0\}$. This proves item (i).

We now prove item (ii). Since all the functions have compact support, without loss of generality, we assume that all $\tilde{\phi}, \phi, \tilde{\psi}, \psi$ are supported inside $[-N_0,N_0]$.
For $N>2N_0$, let $\eta_N$ be a compactly supported $C^\infty$ function such that $\eta_N(x)=1$ for $x\in [-N,N]$. Let $\pp\in \PL_{m-1}$. Then $\eta_N \pp\in \Lp{2}$.
Now it follows directly from the representation in \eqref{df:repr} that
\[
\eta_N \pp=
\sum_{k\in \Z} \la \eta_N \pp, \tilde{\phi}(\cdot-k)\ra \phi(x-k)
+\sum_{j=0}^\infty \sum_{k\in \Z} \la \eta_N \pp, \tilde{\psi}_{j;k}\ra \psi_{j;k}.
\]
Due to the vanishing moments of $\tilde{\psi}$, if the support of $\tilde{\psi}_{j;k}$ is contained inside $[-N,N]$, then $\la \eta_N \pp, \tilde{\psi}_{j;k}\ra=\la \pp, \tilde{\psi}_{j;k}\ra=0$.
Since $\psi$ and $\tilde{\psi}$ are supported inside $[-N_0,N_0]$, the supports of $\tilde{\psi}_{j;k}$ and $\psi_{j;k}$ are contained inside $[2^{-j}(k-N_0),2^{-j}(k+N_0)]$.
Now for $k\in [N_0-2^j N, 2^j N-N_0]$,
we have
\[
2^{-j}(k+N_0)\le 2^{-j}(2^jN-N_0+N_0)=N,\qquad
2^{-j}(k-N_0) \ge 2^{-j} (N_0-2^j N-N_0)=-N.
\]
That is, for every $k\in \Z\cap
[N_0-2^j N, 2^j N-N_0]$,
the support of $\tilde{\psi}_{j;k}$ is contained
inside $[-N,N]$ and hence
$\la \eta_N \pp, \tilde{\psi}_{j;k}\ra=\la \pp, \tilde{\psi}_{j;k}\ra=0$.
On the other hand, for all integers $k< N_0-2^jN$, we must have
\[
2^{-j}(k+N_0)<2^{-j} (N_0-2^j N+N_0)=2^{1-j}N_0-N\le 2N_0-N,
\]
where we used $j\in \N \cup\{0\}$.
Hence, $\psi_{j;k}(x)=0$ for all $x\in (2N_0-N,N-2N_0)$ and $k< N_0-2^jN$.
Similarly, for all integers $k>2^jN-N_0$, we have
\[
2^{-j}(k-N_0)>2^{-j} (2^j N-N_0-N_0)=N-2^{1-j}N_0 \ge N-2N_0.
\]
Hence, $\psi_{j;k}(x)=0$ for all $x\in (2N_0-N,N-2N_0)$ and $k>2^jN-N_0$.
The discussion for the term $\la \eta_N \pp, \tilde{\phi}(\cdot-k)\ra \phi(x-k)$ is the same as the case for $\psi_{j,k}$ with $j=0$. Consequently, we conclude that
\[
\pp(x)=\eta_N (x) \pp(x)=\sum_{k\in \Z} \la \pp, \tilde{\phi}(\cdot-k)\ra \phi(x-k),\qquad a.e.\; x\in (2N_0-N, N-2N_0).
\]
Taking $N\to \infty$, we conclude that $\qp \pp=\pp$. This proves item (ii).
\ep

We have the following result on the Gibbs phenomenon of wavelet and framelet expansions.

\begin{theorem}\label{thm:framelet:gibbs}
Let $(\{\tilde{\phi};\tilde{\psi}\},\{\phi;\psi\})$ be a dual framelet in $\Lp{2}$, where $\phi, \tilde{\phi}$ are $r\times 1$ vectors of compactly supported functions in $\Lp{2}$ and $\psi, \tilde{\psi}$ are $s\times 1$ vectors of compactly supported functions in $\Lp{2}$. Let $\mathcal{A}_n, n\in \N\cup\{0\}$ be defined in \eqref{df:repr:trunc} for the truncated framelet expansions using the dual framelet $(\{\tilde{\phi};\tilde{\psi}\},\{\phi;\psi\})$.
Define the quasi-projection operator $\qp$ as in \eqref{qp:0}.
\begin{enumerate}
\item[(i)] If $\vmo(\psi)\ge 2$ and $\vmo(\tilde{\psi})\ge 1$, then $\la (\cdot), \sgn-\qp \sgn\ra=0$. If in addition \eqref{sgn:cond} holds and $\phi,\tilde{\phi}$ are real-valued, then $\{\mathcal{A}_n\}_{n\in \N}$ exhibits the Gibbs phenomenon at the origin.
\item[(ii)] If all the entries in $\phi$ and $\tilde{\phi}$ are nonnegative, then $\{\mathcal{A}_n\}_{n\in \N}$ has no Gibbs phenomenon at the origin but $\vmo(\psi)=\vmo(\tilde{\psi})=1$.
\end{enumerate}
\end{theorem}

\bp By item (i) of Proposition~\ref{prop:framelet}, the operators $\mathcal{A}_n$ in \eqref{df:repr:trunc} agree with the quasi-projection operators $\qp_n$ in \eqref{qp}.
Since $\vmo(\tilde{\psi})\ge 1$, by item (ii) of Proposition~\ref{prop:framelet}, we have $\qp 1=1$, which is equivalent to \eqref{basic:cond}. Define $\tilde{\qp}_n$ as in \eqref{qp:dual} and $\tilde{\qp}:=\tilde{\qp}_0$. Since $\vmo(\psi)\ge 2$, by item (ii) of Proposition~\ref{prop:framelet}, we have $\tilde{\qp} 1=1$ and $\tilde{\qp} x=x$, which is equivalent to \eqref{basic:cond:dual}.
Next, we prove $[\ol{\wh{\phi}}^\tp \wh{\tilde{\phi}}]''(0)=0$.
As we have seen in the proof of Proposition~\ref{prop:framelet}, \eqref{cascade} holds.
Applying \cite[Lemma~5]{han10} to \eqref{cascade} with $n=1$ and $\gl=1/2$, we conclude that
\[
\ol{\wh{\phi}(\xi)}^\tp \wh{\tilde{\phi}}(\xi)+
\ol{\wh{\psi}(\xi)}^\tp \wh{\tilde{\psi}}(\xi)=
\ol{\wh{\phi}(\xi/2)}^\tp \wh{\tilde{\phi}}(\xi/2).
\]
Because $\vmo(\psi)\ge 2$ and $\vmo(\tilde{\psi})\ge 1$,
we have $\ol{\wh{\psi}(\xi)}^\tp \wh{\tilde{\psi}}(\xi)=\bo(|\xi|^3)$ as $\xi \to 0$. Therefore, we conclude from the above identity that
\be \label{phi:psi:identity}
\ol{\wh{\phi}(\xi)}^\tp \wh{\tilde{\phi}}(\xi)=
\ol{\wh{\phi}(\xi/2)}^\tp \wh{\tilde{\phi}}(\xi/2)+\bo(|\xi|^3),\qquad \xi\to 0.
\ee
By \eqref{basic:cond:dual}, we trivially have $\ol{\wh{\phi}(0)}^\tp \wh{\tilde{\phi}}(0)=1$. Using the Taylor expansion of $\ol{\wh{\phi}}^\tp \wh{\tilde{\phi}}$ at the origin, we deduce from \eqref{phi:psi:identity} that we must have
$\ol{\wh{\phi}(\xi)}^\tp \wh{\tilde{\phi}}(\xi)=1+\bo(|\xi|^3)$ as $\xi \to 0$. Consequently, we proved $[\ol{\wh{\phi}}^\tp \wh{\tilde{\phi}}]''(0)=0$. Now all the claims in item (i) follows directly from Theorem~\ref{thm:qp}.

If all the entries in $\phi$ and $\tilde{\phi}$ are nonnegative, then it follows from item (ii) of Proposition~\ref{prop:nogibb} that $\{\qp_n\}_{n\in \N}$ has no Gibbs phenomenon.
We now prove $\vmo(\psi)=\vmo(\tilde{\psi})=1$. Since
$(\{\tilde{\phi};\tilde{\psi}\},\{\phi;\psi\})$ is a dual framelet in $\Lp{2}$, it is necessary that $\vmo(\psi)\ge 1$ and $\vmo(\tilde{\psi})\ge 1$. We use proof by contradiction to prove $\vmo(\psi)=\vmo(\tilde{\psi})=1$. Without loss of generality, we can assume that $\vmo(\psi)\ge 2$.
By what has been proved a moment ago, we must have $[\ol{\wh{\phi}}^\tp \wh{\tilde{\phi}}]''(0)=0$.
Define $\eta(x):=\int_{\R} \ol{\phi(x+y)}^\tp \tilde{\phi}(y) dy$. Then $\wh{\eta}(\xi)=\ol{\wh{\phi}(\xi)}^\tp \wh{\tilde{\phi}}(\xi)$. Hence, $[\ol{\wh{\phi}}^\tp \wh{\tilde{\phi}}]''(0)=0$ becomes $[\wh{\eta}]''(0)=0$, which is equivalent to saying that $\int_{\R} x^2 \eta(x) dx=0$. Since both $\phi$ and $\tilde{\phi}$ are nonnegative, the scalar function $\eta$ is nonnegative. Consequently, $\int_{\R} x^2 \eta(x) dx=0$ forces $\eta=0$, a contradiction to $\wh{\eta}(0)=1$. This proves $\vmo(\psi)=\vmo(\tilde{\psi})=1$.
This completes the proof of item (ii).
\ep

Note that an orthogonal wavelet is a special case of a tight framelet.
As a direct consequence of Theorem~\ref{thm:framelet:gibbs}, we have
the following result on tight framelets and orthogonal wavelets.

\begin{cor}\label{cor:tf}
Let $\{\phi;\psi\}$ be a tight framelet in $\Lp{2}$, where $\phi$ is an $r\times 1$ vector of compactly supported functions in $\Lp{2}$ and $\psi$ is an $s\times 1$ vector of compactly supported functions in $\Lp{2}$. Let $\mathcal{A}_n, n\in \N$ be the operators for the truncated framelet expansions using the tight framelet $\{\phi;\psi\}$, i.e.,
\be \label{qp:phi:2}
\mathcal{A}_n f:=\sum_{k\in \Z} \la f, \phi(\cdot-k)\ra \phi(\cdot-k)+\sum_{j=0}^{n-1} \sum_{k\in \Z} \la f, \psi_{j;k}\ra \psi_{j;k}.
\ee
If $\vmo(\psi)\ge 2$, then $\la (\cdot), \sgn-\qp \sgn\ra=0$, where $\qp f:=\sum_{k\in \Z} \la f, \phi(\cdot-k)\ra \phi(\cdot-k)$. If in addition \eqref{sgn:cond} holds and $\phi$ is real-valued, then $\{\mathcal{A}_n\}_{n\in \N}$ must exhibit the Gibbs phenomenon at the origin.
\end{cor}

Theorem~\ref{thm:framelet:gibbs} and Corollary~\ref{cor:tf} obviously cover the results in \cite{kel96,ml18,rf05,shen02,sv96} as special cases on the Gibbs phenomenon of orthogonal wavelets and biorthogonal wavelets.
Corollary~\ref{cor:tf} tells us that the truncated approximation using a tight framelet cannot avoid the Gibbs phenomenon while having at least two vanishing moments. However, high vanishing moments are of paramount importance in many applications. Therefore, to avoid the Gibbs phenomenon while having high vanishing moments, according to Theorem~\ref{thm:framelet:gibbs}, the only possibility left is dual framelets $(\{\tilde{\phi};\tilde{\psi}\},\{\phi;\psi\})$ (which must not be tight framelets, i.e., dual framelets with $\tilde{\phi}=\phi$ and $\tilde{\psi}=\psi$)
such that $\vmo(\psi)=1$ and $\vmo(\tilde{\psi})=m$ for $m\ge 2$. By Proposition~\ref{prop:framelet}, the associated quasi-projection operators $\{\qp_n\}_{n\in \N}$ will have accuracy order $m$. To our best knowledge, so far there is no known example of compactly supported dual framelets $(\{\tilde{\phi};\tilde{\psi}\},\{\phi;\psi\})$ in $\Lp{2}$ such that its associated truncation operators $\{\mathcal{A}_n\}_{n\in \N}$ exhibit no Gibbs phenomenon and the framelet functions satisfy
$\vmo(\psi)=1$ and $\vmo(\tilde{\psi})\ge 2$.
It is very interesting to construct such examples and explore their applications in practice.

\section{Proof of Theorem~\ref{thm:main}}
\label{sec:proof}

In this section we shall prove Theorem~\ref{thm:main}.
To do so, we need some necessary definitions and auxiliary results.

By $\lp{0}$ we denote the space of all finitely supported complex-valued sequences $u=\{u(k)\}_{k\in \Z}$ on $\Z$ such that $u(k)\ne 0$ for finitely many $k\in \Z$.
For a finitely supported sequence $u=\{u(k)\}_{k\in \Z}\in \lp{0}$, its Fourier series is defined to be $\wh{u}(\xi):=\sum_{k\in \Z} u(k) e^{-ik\xi}$ for $\xi\in \R$. Note that $\wh{u}(0)=\sum_{k\in\Z} u(k)$. For two finitely supported sequences $u$ and $d$, their convolution is defined to be $[u*d](n)=\sum_{k\in \Z} u(n-k) d(n)$ for $n\in \Z$.
Note that $\wh{u*d}(\xi)=\wh{u}(\xi)\wh{d}(\xi)$.

To prove Theorem~\ref{thm:main}, we need the following auxiliary result.

\begin{lemma}\label{lem:c*d}
Define a sequence $v$ on $\Z$ such that $v(k)=-1$ for all $k<0$ and $v(k)=1$ for all $k\ge 0$.
Let $c$ be a sequence on $\Z$ such that $c-c_\infty v$ is a finitely supported sequence in $\lp{0}$ for some $c_\infty\in \C$.
Let $d\in \lp{0}$ be a finitely supported sequence on $\Z$ such that $\wh{d}(0)=0$.
Then $c*d\in \lp{0}$ is a finitely supported sequence on $\Z$ and
\begin{align}
&\sum_{k\in \Z}[c*d](k)=c_\infty\sum_{k\in \Z} [v*d](k)=-2i c_\infty [\wh{d}]'(0),\label{ord0}\\
&\sum_{k\in \Z} k[c*d](k)=ic_\infty [\wh{d}]'(0)+c_\infty [\wh{d}]''(0)+i\wh{[c-c_\infty v]}(0) [\wh{d}]'(0).\label{ord1}
\end{align}
%
%where $v$ is the sequence such that $v(k)=-1$ for all $k<0$ and $v(k)=1$ for all $k\ge 0$.
\end{lemma}

\bp By our assumption on the sequence $c$, since $d$ is finitely supported with $\wh{d}(0)=\sum_{k\in \Z} d(k)=0$, it is not difficult to observe that $c*d$ is a finitely supported sequence.
Write $c*d=(c-c_\infty v)*d+c_\infty v*d$. Since $c-c_\infty v$ is a finitely supported sequence, we have
$\sum_{k\in\Z} [(c-c_\infty v)*d](k)=\wh{[c-c_\infty v]}(0)\wh{d}(0)=0$ by $\wh{d}(0)=0$.
Define a sequence $u$ such that $u(k)=1$ for all $k\ge 0$ and $u(k)=0$ for all $k<0$. Then $v=-1+2u$. Since $\wh{d}(0)=0$, we have $1*d=0$ and hence $v*d=-1*d+2u*d=2u*d$.
Since $d$ is finitely supported,
we can assume that $d$ is supported inside $[-N,N]$ for some $N\in \N$. Note that $\sum_{n=-N}^N d(n)=\wh{d}(0)=0$. Then
\begin{align*}
\sum_{k\in \Z} 2[u*d](k)
&=2\sum_{k\in \Z} \sum_{n\in \Z} u(n) d(k-n)
=2\sum_{k\in \Z} \sum_{n=0}^\infty d(k-n)
=2\sum_{k\in \Z} \sum_{n=-\infty}^k d(n) \\
&=2\sum_{k\in \Z} \sum_{n=-N}^k d(n)
=2\sum_{k=-N}^N \sum_{n=-N}^k d(n)
=2\sum_{n=-N}^N (N+1-n) d(n)\\
&= -2\sum_{n=-N}^N n d(n)
=-2i[\wh{d}]'(0),
\end{align*}
where we used $[\wh{d}]'(0)=\sum_{n\in \Z} d(n) (-in)$.
This proves \eqref{ord0}.

By $v*d=2u*d$, we have $c*d=(c-c_\infty v)*d+c_\infty v*d=2c_\infty u*d+(c-c_\infty v)*d$. Noting that
$[u*d](k)=\sum_{n=-\infty}^k d(n)$, we deduce that
\[
\sum_{k\in \Z} k[c*d](k)=
2c_\infty \sum_{k\in \Z}\sum_{n=-\infty}^k k  d(n)+\sum_{k\in \Z} k [(c-c_\infty v)*d](k).
\]
Since both $c-c_\infty v$ and $d$ are finitely supported sequences, by $\wh{d}(0)=0$, we have
\[
\sum_{k\in \Z} k [(c-c_\infty v)*d](k)=i[\wh{(c-c_\infty v)}\wh{d}]'(0)=
i\wh{[c-c_\infty v]}(0)[\wh{d}]'(0)+i\wh{[c-c_\infty v]}'(0)\wh{d}(0)=
i\wh{[c-c_\infty v]}(0)[\wh{d}]'(0).
\]
On the other hand, by $\sum_{n=-N}^N d(n)=\wh{d}(0)=0$, we have
\begin{align*}
2\sum_{k\in \Z} \sum_{n=-\infty}^k k d(n)
&=2\sum_{k=-N}^N \sum_{n=-N}^k kd(n)
=2\sum_{n=-N}^N \sum_{k=n}^N kd(n)
=2\sum_{n=-N}^N \frac{N+n}{2}(N+1-n) d(n)\\
&=\sum_{n=-N}^N (N(N+1)+n-n^2) d(n)
=\sum_{n=-N}^N nd(n)-\sum_{n=-N}^N n^2 d(n)
=i[\wh{d}]'(0)+[\wh{d}]''(0).
\end{align*}
Putting all the identities together,
we proved \eqref{ord1}.
\ep

We are now ready to prove Theorem~\ref{thm:main}.

\begin{proof}[Proof of Theorem~\ref{thm:main}]
We first take a special compactly supported scalar function $\eta\in \Lp{2}$ such that
\be \label{sp:eta}
\wh{\eta}(\xi)=1+\bo(|\xi|^3) \quad \mbox{and}\quad
\ol{\wh{\eta}(\xi)} \wh{\eta}(\xi+2\pi k)=\td(k)+\bo(|\xi|^3),\qquad \xi \to 0, k\in \Z.
\ee
Such a function can be easily constructed from B-splines in \eqref{bsplines}. Indeed, consider $B_3$ and note $\wh{B_3}(0)=1$.
Let $u=\{u(k)\}_{k\in \Z}$ be a finitely supported sequence such that $\wh{u}(\xi)=\frac{1}{\wh{B_3}(\xi)}+\bo(|\xi|^3)$ as $\xi \to 0$.
Define $\eta:=\sum_{k\in \Z} u(k)B_3(\cdot-k)$. Then $\eta$ is a compactly supported function in $\Lp{2}$ and
\[
\wh{\eta}(\xi)=\wh{u}(\xi)\wh{B_3}(\xi)=1+\bo(|\xi|^3),\qquad \xi\to 0.
\]
Since $\wh{B_3}(\xi+2\pi k)=\bo(|\xi|^3)$ as $\xi \to 0$ for all $k\in \Z\bs\{0\}$, the function $\eta$ satisfies all the conditions in \eqref{sp:eta}.
By \eqref{sp:eta}, we have
$\la (\cdot)^j, \eta(\cdot-k)\ra=k^j$ and
$x^j=\sum_{k\in \Z} \la (\cdot)^j, \eta(\cdot-k)\ra \eta(x-k)=\sum_{k\in \Z} k^j \eta(\cdot-k)$ for all $j=0,1,2$.
That is, we have
\be \label{eta}
1=\sum_{k\in \Z} \eta(\cdot-k),\qquad
x=\sum_{k\in \Z} k\eta(x-k),\qquad
x^2=\sum_{k\in \Z}k^2 \eta(x-k).
\ee
%
%For simplicity, we define $f:=\sgn$.
By our assumption in \eqref{basic:cond}, we have $\qp 1=1$. Therefore, $\qp \sgn$ agrees with $\sgn$ outside some neighborhood of the origin. Consequently, $\sgn-\qp \sgn$ is a compactly supported function in $\Lp{2}$.
We now calculate $\la \sgn-\qp \sgn, (\cdot)\ra$ using \eqref{eta}. Using the expression for the polynomial $x$ in \eqref{eta} and noting that $\eta$ has compact support, we have
\[
\la \sgn-\qp \sgn,(\cdot)\ra=\sum_{k\in \Z} k \la \sgn-\qp \sgn,\eta(\cdot-k)\ra=\sum_{k\in \Z} k\Big(\la \sgn, \eta(\cdot-k)\ra-\la \qp \sgn, \eta(\cdot-k)\ra\Big).
\]
On the other hand, by the definition of the quasi-projection operator $\qp$ in \eqref{qp:0}, we have
\[
\la \qp \sgn, \eta(\cdot-k)\ra=\sum_{n\in \Z} \la \sgn,\tilde{\phi}(\cdot-n)\ra \la \phi(\cdot-n), \eta(\cdot-k)\ra=\sum_{n\in \Z} \la \sgn,\tilde{\phi}(\cdot-n)\ra \la \phi, \eta(\cdot-(k-n))\ra.
\]
For $k\in \Z$, define
\be \label{bc}
c(k):=\la \sgn,\eta(\cdot-k)\ra\in \C,\quad \tilde{b}(k):=\la \sgn, \tilde{\phi}(\cdot-k)\ra\in \C^{1\times r},\quad b(k):=\la \phi,\eta(\cdot-k)\ra\in \C^r.
\ee
Then we have
\be \label{eq1}
\la \sgn-\qp \sgn, (\cdot)\ra=\sum_{k\in \Z} k\Big(c(k)-\sum_{n\in \Z} \tilde{b}(n)b(k-n)\Big)=\sum_{k\in \Z} k[c-\tilde{b}*b](k).
\ee
Since both $\phi$ and $\eta$ have compact support, the sequence $b$ must be finitely supported. Define $d:=b-\wh{\phi}(0)\td$. Then $d\in \lp{0}$.
Noting that $\tilde{b}*\td=\tilde{b}$, we can write
\[
c-\tilde{b}*b=(c-\tilde{b}*(\wh{\phi}(0)\td))-
\tilde{b}*(b-\wh{\phi}(0)\td)
=(c-\tilde{b}\wh{\phi}(0))-\tilde{b}*d.
\]

Since both $\eta$ and $\tilde{\phi}$ have compact support, by $\wh{\eta}(0)=1$, there exists a positive constant $N$ such that
$c(k)=\sgn (k)$ and $\tilde{b}(k)=\sgn(k) \ol{\wh{\tilde{\phi}}(0)}^\tp$ for all $|k|\ge N$.
Because $\ol{\wh{\tilde{\phi}}(0)}^\tp \wh{\phi}(0)=1$ by \eqref{basic:cond}, we conclude that $c-\tilde{b}\wh{\phi}(0)$ is a finitely supported sequence.
By the first identity in \eqref{eta} and the definition of the sequence $b$, we have
\[
\wh{b}(0)=\sum_{k\in\Z} b(k)=\Big\la \phi, \sum_{k\in \Z} \eta(\cdot-k)\Big\ra=\la \phi,1\ra=\wh{\phi}(0).
\]
Hence, $\wh{d}(0)=\wh{b}(0)-\wh{\phi}(0)=0$.
Since $\tilde{b}(k)=\sgn(k) \ol{\wh{\tilde{\phi}}(0)}^\tp$ for all $|k|\ge N$, we deduce that $\tilde{b}*d$ is a finitely supported sequence.
By \eqref{eq1}, since $c-\tilde{b}\wh{\phi}(0)\in \lp{0}$ and $\tilde{b}*d\in \lp{0}$, we conclude that
\be \label{eqI1I2}
\la \sgn-\qp \sgn, (\cdot)\ra=I_1-I_2 \quad \mbox{with}\quad
I_1:=\sum_{k\in \Z} k\Big
(c(k)-\tilde{b}(k)\wh{\phi}(0)\Big),
\quad
I_2:=\sum_{k\in \Z} k [\tilde{b}*d](k).
\ee
We now calculate $I_1$ and $I_2$.
To calculate $I_1$, we define
\[
\mathring{c}(k):=\int_0^1 \eta(x-k) dx=\int_{-k}^{1-k} \eta(x) dx,\qquad
\mathring{b}(k):=\int_0^1 \tilde{\phi}(x-k)dx=\int_{-k}^{1-k} \tilde{\phi}(x) dx,
\qquad k\in \Z.
\]
Then it is easy to deduce that $\sum_{k\in \Z} \mathring{b}(k)=\wh{\tilde{\phi}}(0)$ and
\be \label{tm:b}
\tilde{b}(k)=\la \sgn, \tilde{\phi}(\cdot-k)\ra
=-\ol{\wh{\tilde{\phi}}(0)}^\tp+2\sum_{n=-\infty}^k \ol{\mathring{b}(n)}^\tp.
\ee
By $\ol{\wh{\tilde{\phi}}(0)}^\tp \wh{\phi}(0)=1$,  the above identity leads to
$\tilde{b}(k)\wh{\phi}(0)=-1+2\sum_{n=-\infty}^k \ol{\mathring{b}(n)}^\tp \wh{\phi}(0)$.
Similarly, we have $c(k)=\la \sgn, \eta(\cdot-k)\ra=-1+2\sum_{n=-\infty}^k \ol{\mathring{c}(n)}$.
Since both $\mathring{b}$ and $\mathring{c}$ are finitely supported,
we can assume that $\ol{\mathring{c}}-\ol{\mathring{b}}^\tp \wh{\phi}(0)$ is supported inside $[-N,N]$ for some $N\in \N$. Then
\[
\sum_{n=-N}^N
\Big(\ol{\mathring{c}(n)}-\ol{\mathring{b}(n)}^\tp \wh{\phi}(0)\Big)
%=\sum_{n\in \Z} \Big(\ol{\mathring{c}(n)}-\ol{\mathring{b}(n)}^\tp \wh{\phi}(0)\Big)
=\sum_{n\in \Z} \ol{\mathring{c}(n)}-\sum_{n\in \Z} \ol{\mathring{b}(n)}^\tp \wh{\phi}(0)
=\ol{\wh{\eta}(0)}- \ol{\wh{\tilde{\phi}}(0)}^\tp
\wh{\phi}(0)=1-1=0
\]
and
\begin{align*}
I_1&=\sum_{k\in \Z} k \Big(c(k)-\tilde{b}(k) \wh{\phi}(0)\Big)
=\sum_{k\in \Z} \sum_{n=-\infty}^k 2 k \Big(\ol{\mathring{c}(n)}-\ol{\mathring{b}(n)}^\tp \wh{\phi}(0)\Big)
=\sum_{k=-N}^N \sum_{n=-N}^k 2 k \Big(\ol{\mathring{c}(n)}-\ol{\mathring{b}(n)}^\tp \wh{\phi}(0)\Big)\\
&=\sum_{n=-N}^N \sum_{k=n}^N 2k \Big(\ol{\mathring{c}(n)}-\ol{\mathring{b}(n)}^\tp \wh{\phi}(0)\Big)
=\sum_{n=-N}^N (N(N+1)+n-n^2) \Big(\ol{\mathring{c}(n)}-\ol{\mathring{b}(n)}^\tp \wh{\phi}(0)\Big)\\
&=\sum_{n\in \Z} (n-n^2)\Big(\ol{\mathring{c}(n)}-\ol{\mathring{b}(n)}^\tp \wh{\phi}(0)\Big).
\end{align*}
By the definition of $\mathring{c}$ and \eqref{eta}, we have
\[
\sum_{n\in \Z} (n-n^2) \mathring{c}(n)
=\sum_{n\in \Z} (n-n^2) \int_{0}^{1} \eta(x-n) dx
=\int_{0}^{1} \sum_{n\in \Z} (n-n^2)\eta(x-n) dx
=\int_0^1 (x-x^2) dx=\frac{1}{6}.
\]
Similarly, by the definition of $\mathring{b}$ and the definition of $\kappa_j$ in \eqref{kappa}, we have
\[
\sum_{n\in \Z} (n-n^2)\mathring{b}(n)
=\int_0^1 \sum_{n\in \Z} (n-n^2) \tilde{\phi}(x-n) dx
=\kappa_1-\kappa_2.
\]
Hence, we conclude that
$I_1=\frac{1}{6}-\ol{(\kappa_1-\kappa_2)}^\tp \wh{\phi}(0)$.

We now calculate $I_2$. Note that we already proved that $\tilde{b}(k)=\sgn(k) \ol{\wh{\tilde{\phi}}(0)}^\tp$ for all $|k|\ge N$.
Define $\tilde{d}:=\tilde{b}-\ol{\wh{\tilde{\phi}}(0)}^\tp v$,
where the sequence $v$ is defined in Lemma~\ref{lem:c*d}.
Since $d\in \lp{0}$ with $\wh{d}(0)=0$,
by the identity \eqref{ord1} in Lemma~\ref{lem:c*d}, we have
\[
I_2=\sum_{k\in \Z} k[\tilde{b}*d](k)=i \ol{\wh{\tilde{\phi}}(0)}^\tp [\wh{d}]'(0)+\ol{\wh{\tilde{\phi}}(0)}^\tp [\wh{d}]''(0)+i\wh{\tilde{d}}(0)[\wh{d}]'(0).
\]
By definition, $d=b-\wh{\phi}(0)\td$. Hence, $\wh{d}(\xi)=\wh{b}(\xi)-\wh{\phi}(0)$. By \eqref{eta} and the definition of the sequence $b$, for $j=0,1,2$, we have
\[
[\wh{b}]^{(j)}(0)
%&=(-i)^j \sum_{k\in \Z}  k^j b(k)
=(-i)^j \sum_{k\in \Z} k^j \la \phi, \eta(\cdot-k)\ra
=(-i)^j \Big\la \phi, \sum_{k\in \Z} k^j \eta(\cdot-k)\Big\ra
=(-i)^j\la \phi, x^j\ra=[\wh{\phi}]^{(j)}(0).
\]
So, $[\wh{d}]'(0)=[\wh{b}]'(0)=[\wh{\phi}]'(0)$ and
 $[\wh{d}]''(0)=[\wh{b}]''(0)=[\wh{\phi}]''(0)$.
Hence,
$I_2=i \ol{\wh{\tilde{\phi}}(0)}^\tp[\wh{\phi}]'(0)
+\ol{\wh{\tilde{\phi}}(0)}^\tp[\wh{\phi}]''(0)
+i{\wh{\tilde{d}}}(0)[\wh{\phi}]'(0)$.
Note that
$\wh{\tilde{d}}(0)
=\sum_{k\in \Z} (\tilde{b}(k)-\ol{\wh{\tilde{\phi}}(0)}^\tp v(k))$. Since $v(k)=-1+2\sum_{n=-\infty}^k \td(n)$
for all $k\in \Z$ and \eqref{tm:b} holds, we deduce that
\begin{align*}
\wh{\tilde{d}}(0)
%=\sum_{k\in \Z} (\tilde{b}(k)-\ol{\wh{\tilde{\phi}}(0)}^\tp v(k))
&=\sum_{k\in \Z} \sum_{n=-\infty}^k 2 \Big(\ol{\mathring{b}(n)}^\tp-
\ol{\wh{\tilde{\phi}}(0)}^\tp \td(n)\Big)
=\sum_{n=-N}^N \sum_{k=n}^N 2 \Big(\ol{\mathring{b}(n)}^\tp-
\ol{\wh{\tilde{\phi}}(0)}^\tp \td(n)\Big)\\
&=\sum_{n=-N}^N 2(N+1-n)\Big(\ol{\mathring{b}(n)}^\tp-
\ol{\wh{\tilde{\phi}}(0)}^\tp \td(n)\Big)
=\sum_{n\in \Z} -2n \Big(\ol{\mathring{b}(n)}^\tp-
\ol{\wh{\tilde{\phi}}(0)}^\tp \td(n)\Big)\\
&=-2\sum_{n\in \Z} n \ol{\mathring{b}(n)}^\tp
=-2\ol{\int_0^1 \Big(\sum_{n\in \Z} n \tilde{\phi}(x-n)\Big)dx}^\tp
=-2\ol{\kappa_1}^\tp.
\end{align*}
Hence,
\[
I_2=i \ol{\wh{\tilde{\phi}}(0)}^\tp[\wh{\phi}]'(0)
+\ol{\wh{\tilde{\phi}}(0)}^\tp[\wh{\phi}]''(0)-2i\ol{\kappa_1}^\tp[\wh{\phi}]'(0).
\]
Therefore, by \eqref{eqI1I2} and $\la (\cdot),\sgn-\qp \sgn\ra=\ol{\la \sgn-\qp \sgn, (\cdot)\ra}$, we conclude that \eqref{identity} holds.
\end{proof}

\section{Gibbs Phenomenon at an Arbitrary Point}
\label{sec:gibbs:x0}

We now discuss the Gibbs phenomenon at an arbitrary point for approximation by framelet expansions and quasi-projection operators.
In particular, we shall deal with the Gibbs phenomenon at an arbitrary point of wavelets and framelets derived from refinable vector functions.

We first discuss the Gibbs phenomenon at an arbitrary point $x_0\in \R$ for approximation by quasi-projection operators.
To do so, we have to introduce some notation.
For a real number $x\in \R$,  $[x]$
stands for its integer part
and $\la x\ra$ stands for its decimal part, i.e.,
\[
x=[x]+\la x\ra, \qquad x\in \R \quad \mbox{with}\quad [x]\in \Z, \; \la x\ra \in [0,1).
\]
Let $\phi$ and $\tilde{\phi}$ be compactly supported vector functions in $(\Lp{2})^r$. For $t\in \R$ and $n\in \N\cup\{0\}$, the \emph{$t$-shifted quasi-projection operators} $\qp_{n,t}$ are defined to be
\be \label{qp:shift}
\qp_{n,t} f:=\sum_{k\in \Z} \la f, 2^n \tilde{\phi}(2^n \cdot-k+t)\ra \phi(2^n\cdot-k+t),\qquad
n\in \N\cup\{0\}, f\in \lLp{2}.
\ee
Note that $\qp_n=\qp_{n,0}$ and
\be \label{qp:0:t}
\qp_{0,t+k}f=\qp_{0,t} f= [\qp (f(\cdot-t))](\cdot+t),\qquad \forall\, t\in \R, k\in \Z.
\ee
It follows directly from the above identity
in \eqref{qp:0:t} that $\qp \pp=\pp$ for all $\pp\in \PL_{m-1}$ if and only if $\qp_{0,t} \pp=\pp$ for all $\pp\in \PL_{m-1}$.

The following result on the shifted quasi-projection operators will be needed in our discussion of the Gibbs phenomenon at an arbitrary point.

\begin{lemma}\label{lem:qp:shift}
Let $\phi$ and $\tilde{\phi}$ be compactly supported vector functions in $(\Lp{2})^r$. Suppose that $\qp1:=\sum_{k\in \Z} \la 1, \tilde{\phi}(\cdot-k)\ra \phi(\cdot-k)=1$ (i.e., \eqref{basic:cond} holds) and all the entries of $\phi$ are continuous. Then
\be \label{qp:continuity}
\lim_{t\to c} \sup_{x\in \R} |[\qp_{0,t} \sgn-\qp_{0,c} \sgn](x)|
=0,\qquad \forall\, c\in \R.
\ee
Let $[c,d]$ be a bounded and closed interval and $t\in \R$. If $f$ is
a square integrable function such that $\lim_{y\to x} f(y)=f(x)$
for all $x\in [c,d]$, then
\be \label{qp:cd}
\lim_{n\to \infty} \sup_{x\in [c,d]} |[\qp_{n,t} f](x)-f(x)|=0.
\ee
\end{lemma}

\bp By \eqref{qp:0:t}, we have $\qp_{0,c}=\qp_{0,c+k}$ for all $k\in \Z$. Hence, without loss of generality, we can assume $c\in [0,1)$. So, we only consider $t\in (-1,1)$ in \eqref{qp:continuity}.
Because $\phi$ and $\tilde{\phi}$ have compact support, without loss of generality, we assume that they are supported inside $[-N,N]$ for some $N\in \N$. By $t \in (-1,1)$, both $\tilde{\phi}(\cdot-k+t)$ and
$\phi(\cdot-k+t)$ are supported inside $[k-N-1,k+N+1]$.
Hence, both of them are supported inside $[0,\infty)$ for all $k\ge N+1$, and inside $(-\infty,0]$ for all $k\le -N-1$. Now
by $\qp1=1$, we have $\qp_{0,t}1=1$ and thus
we can easily deduce that
\be \label{qp:supp}
[\qp_{0,t} \sgn](x)=\sgn(x) \qquad \forall\; x \not\in [-2N-1,2N+1],\; t\in (-1,1).
\ee
Since $c\in [0,1)$, the identity in \eqref{qp:supp} holds with $t=c$. Consequently, since $\phi$ is continuous,
\[
\| \qp_{0,t} \sgn -\qp_{0,c} \sgn\|_{C(\R)}:=
\mbox{sup}_{x\in \R}|[\qp_{0,t} \sgn -\qp_{0,c} \sgn](x)|
=
\mbox{sup}_{x\in [-2N-1,2N+1]}|[\qp_{0,t} \sgn -\qp_{0,c} \sgn](x)|.
\]
For $x\in [-2N-1,2N+1]$, we observe
\[
[\qp_{0,t} \sgn -\qp_{0,c} \sgn](x)
=
\sum_{k=-3N-2}^{3N+2} H_{t,c,k}(x),
\]
where
\[
H_{t,c,k}(x):=
\la \sgn, \tilde{\phi}(\cdot-k+t)\ra \phi(x-k+t)-
\la \sgn, \tilde{\phi}(\cdot-k+c)\ra \phi(x-k+c).
\]
Therefore,
\be \label{qp:H}
\| \qp_{0,t} \sgn -\qp_{0,c} \sgn\|_{C(\R)}\le \sum_{k=-3N-2}^{3N+2}
\|H_{t,c,k}\|_{C(\R)}.
\ee
By $\lim_{t\to c} \| \tilde{\phi}(\cdot+t)-\tilde{\phi}(\cdot+c)\|_{(\Lp{2})^r}=0$, we trivially have $\lim_{t\to c} \la \sgn, \tilde{\phi}(\cdot-k+t)\ra=\la \sgn, \tilde{\phi}(\cdot-k+c)\ra$.
On the other hand, since $\phi$ is a compactly supported continuous vector function, the vector function $\phi$ must be bounded and uniformly continuous on $\R$. Hence $\lim_{t\to c}
\|\phi(\cdot+t)-\phi(\cdot+c)\|_{(C(\R))^r}=0$.
Now it is straightforward to conclude that $\lim_{t\to c} \|H_{t,c,k}\|_{C(\R)}=0$. Consequently, it follows directly from \eqref{qp:H} that \eqref{qp:continuity} must hold.

We now prove \eqref{qp:cd}. Let $\gep>0$. Since $f$ is continuous on the compact set $[c,d]$, the function $f$ must be uniformly continuous on $[c,d]$. Together with the fact that $f$ is continuous at every point in $[c,d]$, there must exist $\delta>0$ such that
\[
|f(y)-f(x)|<\gep/2\qquad \forall\; x,y\in [c,d], |y-x|<\delta
\]
and
\[
|f(y)-f(x)|<\gep/2\qquad \forall\;  x\in \{c,d\}, |y-x|<\delta.
\]
Then we can deduce from the above inequalities that
\be \label{f:cd}
|f(y)-f(x)|<\gep \qquad \forall\; x\in [c,d], |y-x|<\delta.
\ee
Since $\qp1=1$ and $\phi$ is continuous,
we must have $[\qp_{n,t}1](x)=1$ for all $x\in \R$. Consequently,
\[
[\qp_{n,t} f](x_0)-f(x_0)=\sum_{k\in \Z} \la f(\cdot)-f(x_0), 2^n \tilde{\phi}(2^n\cdot-k+t)\ra \phi(2^n x_0-k+t),\qquad \forall\, x_0\in \R.
\]
Because $\phi$ and $\tilde{\phi}$ are supported inside $[-N,N]$, the supports of $\phi(2^n\cdot-k+t)$ and $\tilde{\phi}(2^n\cdot-k+t)$ are contained inside an interval of length at most $2^{1-n}N$.
Let $x_0\in [c,d]$ be temporarily fixed.
If $\phi(2^n x_0-k+t)\ne 0$, then $k\in [2^n x_0+t-N, 2^n x_0+t+N]$ and hence the support of $\tilde{\phi}(2^n \cdot-k+t)$ must be contained inside the interval $[x_0-2^{1-n}N, x_0+2^{1-n} N]$. For $n>2+\log_2 (\delta/N)$, we have $[x_0-2^{1-n}N, x_0+2^{1-n} N]\subset (x_0-\delta,x_0+\delta)$ and therefore,
\[
[\qp_{n,t} f](x_0)-f(x_0)=\sum_{k\in \Z \cap [2^n x_0+t-N, 2^n x_0+t+N]} \la f(\cdot)-f(x_0), 2^n \tilde{\phi}(2^n\cdot-k+t)\ra \phi(2^n x_0-k+t).
\]
Therefore, for every $k\in \Z \cap [2^n x_0+t-N, 2^n x_0+t+N]$, since the support of $\tilde{\phi}(2^n\cdot-k+t)$ is contained inside
$(x_0-\delta,x_0+\delta)$, using \eqref{f:cd}, we must have
\begin{align*}
\|\la f(\cdot)-f(x_0), 2^n \tilde{\phi}(2^n \cdot-k+t)\ra\|_{l_2}
&\le
\int_{x_0-\delta}^{x_0+\delta}
|f(y)-f(x_0)| 2^n \| \tilde{\phi}(2^n y-k+t)\|_{l_2} dy\\
&\le \gep \int_\R 2^n \| \tilde{\phi}(2^n y-k+t)\|_{l_2} dy
=C_1\gep,
\end{align*}
where $C_1:=\int_{\R} \|\tilde{\phi}(y)\|_{l_2} dy$. Since $\tilde{\phi}\in (\Lp{2})^r$ has compact support, by the Cauchy-Schwarz inequality, we see that $\tilde{\phi}\in (\Lp{1})^r$ and hence $C_1<\infty$. Consequently, we have
\[
|[\qp_{n,t} f](x_0)-f(x_0)|\le
%\sum_{k\in \Z \cap [2^n x_0+t-N, 2^n x_0+t+N]} C_1 \gep  \|\phi(2^n x_0-k+t)\|_{l_2}\le
C_1 \gep \sum_{k\in \Z}  \|\phi(2^n x_0-k+t)\|_{l_2}
\le C_1 C_2 \gep,
\]
where
\[
C_2:=\left\| \sum_{k\in \Z} \|\phi(\cdot-k+t)\|_{l_2}\right\|_{C(\R)}:=
\sup_{x\in \R} \sum_{k\in \Z} \|\phi(x-k+t)\|_{l_2}.
\]
Because $\phi$ is compactly supported and is continuous, $\phi$ is bounded and $C_2<\infty$. This proves
\[
\sup_{x_0\in [c,d]} |[\qp_{n,t} f](x_0)-f(x_0)|\le C_1 C_2 \gep, \qquad \forall\, n> 2+\log_2 (\delta/N).
\]
This completes the proof of \eqref{qp:cd}.
\ep

We now discuss the Gibbs phenomenon at an arbitrary point $x_0\in \R$ for approximation by quasi-projection operators.
Consider the function $\sgn(\cdot-x_0)$, which is continuous everywhere except at the point $x_0$. Noting that $\sgn(2^{-n}\cdot)=\sgn$, we have
\begin{align*}
[\qp_n (\sgn(\cdot-x_0))](x)
&=\sum_{k\in \Z} \la \sgn(\cdot-x_0), 2^n \tilde{\phi}(2^n\cdot-k)\ra \phi(2^n x-k)\\
&=\sum_{k\in \Z}
\la \sgn, \tilde{\phi}(\cdot-k+2^n x_0)\ra \phi(2^n x-k)\\
&=\sum_{k\in \Z}
\la \sgn, \tilde{\phi}(\cdot-k+\la 2^n x_0\ra\ra\phi(2^n x-[2^nx_0]-k).
\end{align*}
Therefore, shifting the function $\qp_n (\sgn(\cdot-x_0))$ back by the amount of $x_0$, we have
\begin{align*}
[\qp_n (\sgn(\cdot-x_0))] (x+x_0)
&=\sum_{k\in\Z} \la \sgn, \tilde{\phi}(\cdot-k+\la 2^n x_0\ra)\ra \phi(2^n(x+x_0)-[2^n x_0]-k)\\
&=\sum_{k\in\Z} \la \sgn, \tilde{\phi}(\cdot-k+\la 2^n x_0\ra)\ra \phi(2^n x-k+\la 2^n x_0\ra)\\
&=[\qp_{0,\la 2^n x_0\ra} \sgn](2^n x).
\end{align*}
Since $\phi$ and $\tilde{\phi}$ are supported inside $[-N,N]$ for some $N\in \N$, as we proved in Lemma~\ref{lem:qp:shift},
\eqref{qp:supp} must hold.
Since $\la 2^n x_0\ra\in [0,1)$,
\eqref{qp:supp} allows us to take $c_n=2^{-n}(2N+1)$ in Definition~\ref{def:gibbs}.
Note that the essential supermum in \eqref{r:gibbs} does not change for any choice of $c_n>2^{-n}(2N+1)$.
So, we define
\be \label{RLx0}
R_{x_0}:=\limsup_{n\to \infty} R(\la 2^n x_0\ra),\qquad
L_{x_0}:=\liminf_{n\to \infty} L(\la 2^n x_0\ra),
\ee
where for $t\in \R$,
\be \label{RL}
R(t):=\mbox{ess-sup}_{x\in (0,\infty)}
[\qp_{0,t} \sgn](x),\qquad
L(t):=\mbox{ess-inf}_{x\in (-\infty,0)}
[\qp_{0,t} \sgn](x).
\ee
Then $\{\qp_n \}_{n\in \N}$ exhibits the Gibbs phenomenon at $x_0$ if and only if either $R_{x_0}>1$ or $L_{x_0}<-1$.
The nonnegative quantities $R_{x_0}-1$ and $-L_{x_0}-1$ are called
the \emph{overshoot percentages} at $x_0$ from the right-hand side and left-hand side, respectively.
For a sequence $\{t_n\}_{n\in \N}$ of real numbers, recall that $c\in \R$ is \emph{a cluster point} of $\{t_n\}_{n\in \N}$ if
there exists a subsequence $\{t_{n_k}\}_{k\in \N}$ such that $\lim_{k\to \infty} t_{n_k}=c$.

Under the assumption that all the entries of $\phi$ are continuous and $\qp1=1$, we now characterize the Gibbs phenomenon at an arbitrary point $x_0\in \R$.

\begin{theorem}\label{thm:gibbsx0}
Let $\phi$ and $\tilde{\phi}$ be compactly supported vector functions in $(\Lp{2})^r$.
For $x_0\in \R$, let $S_{x_0}$ be the set of all the cluster points of the sequence $\{\la 2^n x_0\ra\}_{n\in \N}$.
Define $R_{x_0}$ and $L_{x_0}$ as in \eqref{RLx0}. If $\qp1:=\sum_{k\in \Z} \la 1, \tilde{\phi}(\cdot-k)\ra \phi(\cdot-k)=1$ (i.e., \eqref{basic:cond} holds) and all the entries of $\phi$ are continuous, then
\be \label{RLx0:c}
R_{x_0}=\sup_{c\in S_{x_0}} R(c)\quad \mbox{and}\quad
L_{x_0}=\inf_{c\in S_{x_0}} L(c).
\ee
\end{theorem}

\bp
By definition of $R_{x_0}$, there exists a subsequence $\{t_k\}_{k\in \N}$ with $t_k:=\la 2^{n_k} x_0\ra$ of the sequence $\{\la 2^n x_0\ra\}_{n\in \N}$ such that $R_{x_0}=\lim_{k\to \infty} R(t_k)$.
Since all $t_k$ lie inside the compact set $[0,1]$, $\{t_k\}_{k\in \N}$ must have a convergent subsequence. Without loss of generality, we assume that $\{t_k\}_{k\in \N}$ itself converges to a point $c\in [0,1]$. Hence, $c\in S_{x_0}$. Since $\qp1=1$ and all the entries of $\phi$ are continuous, by Lemma~\ref{lem:qp:shift}, \eqref{qp:continuity} must hold. In particular, \eqref{qp:continuity} implies
$\lim_{k\to \infty} R(t_k)=R(c)$. This proves
\[
R_{x_0}=\lim_{k\to \infty} R(t_k)
=R(c) \le \sup_{t\in S_{x_0}} R(t).
\]
On the other hand, for $c\in S_{x_0}$, the point $c$ is a cluster point of $\{\la 2^n x_0\ra\}_{n\in \N}$ and hence
there exists a subsequence $\{t_k\}_{k\in \N}$ with $t_k:=\la 2^{n_k} x_0\ra$ such that $c=\lim_{k\to \infty} t_k$.
Now by \eqref{qp:continuity}, we have
\[
R(c)=\lim_{k\to \infty} R(t_k)\le \limsup_{n\to \infty} R(\la 2^n x_0\ra)=R_{x_0}.
\]
Hence, $\sup_{c\in S_{x_0}} R(c)\le R_{x_0}$. This proves the first identity in \eqref{RLx0:c}.
The proof of the second identity in \eqref{RLx0:c} is similar.
\ep

Note that $R(t)\ge 1$ and $L(t)\le -1$ for all $t\in \R$. By Theorem~\ref{thm:gibbsx0},
$\{\qp_n \}_{n\in \N}$ does not exhibit the Gibbs phenomenon at $x_0$ if and only if
$R_{x_0}=1$ and $L_{x_0}=-1$, which are further equivalent to

\be \label{qp:gibbs}
\mbox{ess-sup}_{x\in (0,\infty)} [\qp_{0,c}\sgn](x)=1
\quad \mbox{and}\quad
\mbox{ess-inf}_{x\in (-\infty,0)}[\qp_{0,c}\sgn](x)=-1, \qquad \forall\; c\in S_{x_0}.
\ee
If $x_0$ is an irrational number, then it is well known that $S_{x_0}=[0,1]$ with $0\in S_{x_0}$. Let $x_0$ be a rational number. Then we can uniquely write $x_0=\frac{p}{2^k q}$ with $p\in \Z$, $k\in \N\cup\{0\}$, and $q$ being a positive odd integer.
If $q=1$, then $x_0$ is just a dyadic rational number and hence $S_{x_0}=S_{\frac{p}{2^k}}=\{0\}$. For this case,
\eqref{qp:gibbs} is simply Lemma~\ref{lem:qp:gibb}.
If $q>1$, then $S_{x_0}=S_{\frac{p}{2^k q}} \subseteq \{\frac{1}{q},\ldots,\frac{q-1}{q}\}$ with  $0,1\not \in S_{x_0}$.

For the Gibbs phenomenon at an arbitrary point of approximation by general quasi-projection operators, as a consequence of Theorem~\ref{thm:gibbsx0},
the following result generalizes Theorem~\ref{thm:qp}.

\begin{theorem}\label{thm:qp:c}
Let $\phi$ and $\tilde{\phi}$ be $r\times 1$ vectors of compactly supported functions in $\Lp{2}$ such that all the entries of $\phi$ are continuous. Let
$\qp_n$ be the quasi-projection operators defined in \eqref{qp} and $\qp:=\qp_0$ in \eqref{qp:0}. Let $\qp_{0,t}, t\in \R$ be the $t$-shifted quasi-projection operator defined in \eqref{qp:shift}.
If \eqref{basic:cond} holds (i.e., $\qp 1=1$) and
\eqref{basic:cond:dual} is satisfied,
then
\be \label{sgn:qp:c}
\la (\cdot) , \sgn-\qp_{0,c} \sgn\ra=-[\ol{\wh{\phi}}^\tp \wh{\tilde{\phi}}]''(0), \qquad \forall\; c\in \R.
\ee
If in addition $[\ol{\wh{\phi}}^\tp \wh{\tilde{\phi}}]''(0)=0$ and both $\phi$ and $\tilde{\phi}$ are real-valued,
then $\{\qp_n \}_{n\in \N}$ exhibits the Gibbs phenomenon at all points in $\R$.
\end{theorem}

\bp Define $\mathring{\phi}:=\phi(\cdot+c)$ and
$\tilde{\mathring{\phi}}:=\tilde{\phi}(\cdot+c)$.
Obviously, by definition we have
\[
\mathring{\qp}_n f:=\sum_{k\in \Z} \la f, 2^n \tilde{\mathring{\phi}}(2^n\cdot-k)\ra \mathring{\phi}(2^n\cdot-k)=
\sum_{k\in \Z} \la f, 2^n \tilde{\phi}(2^n\cdot-k+c)\ra \phi(2^n\cdot-k+c)
=\qp_{n,c} f.
\]
By $\wh{\mathring{\phi}}(\xi)=e^{ic\xi} \wh{\phi}(\xi)$ and $\wh{\tilde{\mathring{\phi}}}(\xi)=e^{ic\xi} \wh{\tilde{\phi}}(\xi)$, we observe
\be \label{identity:mathring}
\ol{\wh{\tilde{\mathring{\phi}}}(\xi+2\pi j)}^\tp \wh{\mathring{\phi}}(\xi+2\pi k)=
e^{i2\pi (k-j) c} \ol{\wh{\tilde{\phi}}(\xi+2\pi j)}^\tp \wh{\phi}(\xi+2\pi k),\qquad \xi\in \R, j,k\in \Z.
\ee
Since \eqref{basic:cond} and
\eqref{basic:cond:dual} are satisfied for $\phi$ and $\tilde{\phi}$, now it is trivial to deduce from \eqref{identity:mathring} that both \eqref{basic:cond} and
\eqref{basic:cond:dual} are satisfied with $\phi$ and $\tilde{\phi}$ being replaced by $\mathring{\phi}$ and $\tilde{\mathring{\phi}}$, respectively. Applying Theorem~\ref{thm:qp} to $\mathring{\phi}$ and $\tilde{\mathring{\phi}}$ and noting $\ol{\wh{\mathring{\phi}}}^\tp(\xi) \wh{\tilde{\mathring{\phi}}}(\xi)=
\ol{\wh{\phi}}^\tp(\xi) \wh{\tilde{\phi}}(\xi)$,
we must have
\[
\la (\cdot), \sgn-\qp_{0,c} \sgn\ra=
\la (\cdot), \sgn -\mathring{\qp}_0 \sgn\ra=-[\ol{\wh{\mathring{\phi}}}^\tp \wh{\tilde{\mathring{\phi}}}]''(0)=
-[\ol{\wh{\phi}}^\tp \wh{\tilde{\phi}}]''(0).
\]
This proves \eqref{sgn:qp:c}. Since all the entries of $\phi$ are continuous, the condition in \eqref{sgn:cond} holds. Consequently, if $[\ol{\wh{\phi}}^\tp \wh{\tilde{\phi}}]''(0)=0$, then by \eqref{sgn:qp:c} we must have
$\la (\cdot), \sgn-\qp_{0,c} \sgn\ra=0$.
Hence, one of the identities in \eqref{qp:gibbs} must fail. Now it follows from Theorem~\ref{thm:gibbsx0} that
$\{\qp_n \}_{n\in \N}$ must exhibit the Gibbs phenomenon at all points $x_0\in \R$ by noting that the set $S_{x_0}$ cannot be empty.
\ep

As a direct consequence of Theorem~\ref{thm:qp:c},
a general version of Corollary~\ref{cor:qp} is as follows.

\begin{cor}\label{cor:qp:c}
Let $\phi$ be an $r\times 1$ vector of compactly supported real-valued functions in $\Lp{2}$ such that all the entries of $\phi$ are continuous.
Let $\qp_n, n\in \N$ be the quasi-projection operators associated with $\phi$ as in \eqref{qp:phi}.
If $\{\qp_n\}_{n\in \N}$ has accuracy order higher than two, then  $\{\qp_n\}_{n\in \N}$ must exhibit the Gibbs phenomenon at all points in $\R$.
\end{cor}

Finally we discuss the Gibbs phenomenon at an arbitrary point for dual framelets, in particular, for those derived from refinable vector functions and multiresolution analysis.
As a direct consequence of Theorems~\ref{thm:qp:c} and~\ref{thm:framelet:gibbs},
we have the following result generalizing Theorem~\ref{thm:framelet:gibbs}.

\begin{theorem}\label{thm:framelet:gibbs:c}
Let $(\{\tilde{\phi};\tilde{\psi}\},\{\phi;\psi\})$ be a compactly supported real-valued dual framelet in $\Lp{2}$ such that all the entries of $\phi$ are continuous.
Let $\mathcal{A}_n, n\in \N\cup\{0\}$ be defined in \eqref{df:repr:trunc} for the truncated dual framelet expansions.
Define the quasi-projection operator $\qp$ as in \eqref{qp:0}.  If $\vmo(\psi)\ge 2$ and $\vmo(\tilde{\psi})\ge 1$, then $\la (\cdot), \sgn-\qp \sgn_{0,c}\ra=0$ for all $c\in \R$ and $\{\mathcal{A}_n\}_{n\in \N}$ must exhibit the Gibbs phenomenon at all points in $\R$.
\end{theorem}

\bp By Proposition~\ref{prop:framelet}, $\vmo(\tilde{\psi})\ge 1$ implies $\qp 1=1$.
Now all the claim follows from the same proof of Theorem~\ref{thm:framelet:gibbs} but using Theorem~\ref{thm:qp:c} instead of Theorem~\ref{thm:qp}.
\ep

As a direct consequence of Theorem~\ref{thm:framelet:gibbs:c},
the following result generalizes Corollary~\ref{cor:tf}
on tight framelets and orthogonal wavelets.

\begin{cor}\label{cor:tf:c}
Let $\{\phi;\psi\}$ be a compactly supported real-valued tight framelet in $\Lp{2}$ such that all the entries of $\phi$ are continuous and real-valued. Let $\mathcal{A}_n, n\in \N\cup\{0\}$ be defined in \eqref{qp:phi:2} for the truncated tight framelet expansions.
If $\vmo(\psi)\ge 2$, then $\la (\cdot), \sgn-\qp_{0,c} \sgn\ra=0$ for all $c\in \R$
and $\{\mathcal{A}_n\}_{n\in \N}$ must exhibit the Gibbs phenomenon at all points in $\R$.
\end{cor}

We finish this paper by discussing the Gibbs phenomenon of dual framelets derived from refinable vector functions and multiresolution analysis.
Let $\phi=(\phi^1,\ldots,\phi^r), \tilde{\phi}=(\tilde{\phi}^1,\ldots,\tilde{\phi}^r)\in (\Lp{2})^r$ and $\psi=(\psi^1,\ldots,\psi^s), \tilde{\psi}=(\tilde{\psi}^1,\ldots,\tilde{\psi}^s)\in (\Lp{2})^s$.
As addressed in \cite{han10,han12}, the underlying systems of a dual framelet $(\{\tilde{\phi};\tilde{\psi}\},\{\phi;\psi\})$ in $\Lp{2}$ are the nonhomogeneous affine systems $\AS_J(\phi;\psi)$, $J\in \Z$ in \eqref{as}.
In the classical theory of wavelets, however \emph{homogeneous affine systems} are often discussed and are defined to be
\[
\AS(\psi)=\{\psi^\ell_{j;k} \setsp j\in \Z, k\in \Z, \ell=1,\ldots, s\}.
\]
If both $\AS(\psi)$ and $\AS(\tilde{\psi})$ are frames in $\Lp{2}$ and satisfy
\[
\la f,g\ra=\sum_{j\in \Z} \sum_{\ell=1}^s \sum_{k\in \Z}
\la f, \tilde{\psi}^\ell_{j;k}\ra \la \psi^\ell_{j;k},g\ra,\qquad \forall\, f,g\in \Lp{2}
\]
with the series converging absolutely, then $(\tilde{\psi},\psi)$ is called a \emph{homogeneous dual framelet} in $\Lp{2}$.
It is known in \cite{han12} that if $(\{\tilde{\phi};\tilde{\psi}\},\{\phi;\psi\})$ is a dual framelet (resp. biorthogonal wavelet) in $\Lp{2}$,
as the limiting systems of $\AS_J(\tilde{\phi};\tilde{\psi})$ and $\AS_J(\phi,\psi)$ as $J\to -\infty$,
then $(\tilde{\psi},\psi)$ must be a homogeneous dual framelet (resp. homogeneous biorthogonal wavelet) in $\Lp{2}$. Furthermore, due to
the identities in \eqref{cascade}
and the known identity $\lim_{j\to -\infty}
\sum_{k\in \Z} \la f, \tilde{\phi}_{j;k}\ra \la \phi_{j;k},g\ra=0$ for all $f,g\in \Lp{2}$, one has
\[
\sum_{j=-\infty}^{n-1} \sum_{k\in \Z} \la f, \tilde{\psi}_{j;k}\ra \psi_{j;k}=\mathcal{A}_n f=\sum_{k\in \Z} \la f, \tilde{\phi}_{n;k}\ra \phi_{n;k},\qquad f\in \Lp{2},
\]
where $\mathcal{A}_n$ is defined in \eqref{df:repr:trunc}. Hence, all our results can be directly applied to homogeneous wavelets and framelets.
For detailed discussion about the relations between homogeneous and nonhomogeneous wavelets and framelets, see \cite{han12,han17,hanbook} and references therein.

The most general method for constructing tight or dual framelets derived from refinable vector functions is the Oblique Extension Principle (OEP), e.g., see \cite{dh04,dhrs03,hm03}. For systematic construction of OEP-based dual framelets and dual multiframelets, see \cite{chs02,dh04,dhrs03,han09,han10,hanbook,hm03} and references therein.

For $a=\{a(k)\}_{k\in \Z}\in (\lp{0})^{r\times s}$, recall that $\wh{a}$ is the $r\times s$ matrix of $2\pi$-periodic trigonometric polynomials defined by
$\wh{a}(\xi):=\sum_{k\in \Z}a(k) e^{-i k\xi}$ for $\xi\in \R$.
Let $\phi\in (\Lp{2})^{r}$ be a compactly supported vector function. We say that $\phi$ is \emph{refinable} with a filter/mask $a\in (\lp{0})^{r\times r}$ if
$\phi=2\sum_{k\in \Z} a(k) \phi(2\cdot-k)$,
which is equivalent to $\wh{\phi}(2\xi)=\wh{a}(\xi)\wh{\phi}(\xi)$.

Now we recall the Oblique Extension Principle stated in \cite[Theorem~6.4.1]{hanbook} (also see \cite{han09,hm03}) for constructing OEP-based dual framelets in $\Lp{2}$ from refinable vector functions.

\begin{theorem} \label{thm:oep} (\cite[Theorem~6.4.1]{hanbook})
Let $a,\tilde{a},\theta,\tilde{\theta}\in (\lp{0})^{r\times r}$ and $b,\tilde{b}\in (\lp{0})^{r\times s}$. Let $\phi, \tilde{\phi}$ be $r\times 1$ vectors of compactly supported functions in $\Lp{2}$ satisfying
\[
\wh{\phi}(2\xi)=\wh{a}(\xi)\wh{\phi}(\xi),\qquad
\wh{\tilde{\phi}}(2\xi)=\wh{\tilde{a}}(\xi)\wh{\tilde{\phi}}(\xi),\qquad \xi\in \R.
\]
Define vector functions $\eta, \tilde{\eta}, \psi, \tilde{\psi}$ by
\[
\wh{\eta}(\xi):=\wh{\theta}(\xi) \wh{\phi}(\xi), \quad
\wh{\tilde{\eta}}(\xi):=\wh{\tilde{\theta}}(\xi)
\wh{\tilde{\phi}}(\xi)
\quad\mbox{and}\quad
\wh{\psi}(\xi):=
\wh{b}(\xi/2)\wh{\phi}(\xi/2),\quad \wh{\tilde{\psi}}(\xi):=
\wh{\tilde{b}}(\xi/2)\wh{\tilde{\phi}}(\xi/2).
\]
Define $\Theta\in (\lp{0})^{r\times r}$ by $\wh{\Theta}(\xi):=[\wh{\tilde{\theta}}(\xi)]^\tp \ol{\wh{\theta}(\xi)}$. If $\wh{\tilde{\phi}}(0)^\tp \wh{\Theta}(0)\ol{\wh{\phi}(0)}=1$ and
$(\{\tilde{a};\tilde{b}\},\{a;b\})_\Theta$ is \emph{an OEP-based dual framelet filter bank}, that is,
\be \label{oep}
\wh{\tilde{a}}(\xi)^\tp \wh{\Theta}(2\xi) \ol{\wh{a}(\xi)}+\wh{\tilde{b}}(\xi)^\tp \ol{\wh{b}(\xi)}=\wh{\Theta}(\xi),\quad
\wh{\tilde{a}}(\xi)^\tp \wh{\Theta}(2\xi) \ol{\wh{a}(\xi+2\pi)}+\wh{\tilde{b}}(\xi)^\tp \ol{\wh{b}(\xi+2\pi)}=0,
\ee
then $(\{\tilde{\eta}; \tilde{\psi}\},\{\eta;\psi\})$ is a dual framelet in $\Lp{2}$. Moreover, the converse direction is also true if $\det \wh{\Theta}$ is not identically zero.
\end{theorem}

The most common choice of $\theta,\tilde{\theta}$ in Theorem~\ref{thm:oep} is $\wh{\theta}(\xi)=\wh{\tilde{\theta}}(\xi)=I_r$, leading to $\eta=\phi$ and $\tilde{\eta}=\tilde{\phi}$.
The main role of $\theta$ and $\tilde{\theta}$ is to increase the vanishing moments of $\psi$ and $\tilde{\psi}$ as well as the accuracy order of their associated quasi-projection operators: $\mathring{\qp}:=\mathring{\qp}_0$ and
\[
[\mathring{\qp}_n f](x):=\sum_{k\in \Z} \la f, 2^n \tilde{\eta}(2^n \cdot-k)\ra \eta(2^n x-k),\qquad n\in \N\cup\{0\}, f\in \lLp{2}.
\]
By Proposition~\ref{prop:framelet}, we have $\mathring{\qp} \pp=\pp$ for all $\pp\in \PL_{m-1}$ with $m:=\vmo(\tilde{\psi})$. Hence, all the results in this paper hold for OEP-based dual framelets.
Moreover, it is known in \cite[Lemma~4.1.11]{hanbook} that
\be \label{qp:mathring}
\mathring{\qp}_n f=\sum_{k\in \Z} \la f, 2^n \tilde{\mathring{\phi}}(2^n\cdot-k)\ra \phi(2^n\cdot-k)=
\sum_{k\in \Z} \la f, 2^n \tilde{\phi}(2^n\cdot-k)\ra \mathring{\phi}(2^n \cdot-k),\qquad n\in \N\cup\{0\}
\ee
with $\wh{\tilde{\mathring{\phi}}}(\xi):=\wh{\Theta}(\xi)^\tp \wh{\tilde{\phi}}(\xi)$ and $\wh{\mathring{\phi}}(\xi)=\ol{\wh{\Theta}(\xi)}\wh{\phi}(\xi)$.
Consequently, if
$(\{\tilde{\eta}; \tilde{\psi}\},\{\eta;\psi\})$ is a dual framelet in $\Lp{2}$, then both
$(\{\tilde{\phi}; \tilde{\psi}\},\{\mathring{\phi};\psi\})$ and $(\{\tilde{\mathring{\phi}}; \tilde{\psi}\},\{\phi;\psi\})$ are dual framelets in $\Lp{2}$.
In particular, if $\vmo(\psi)\ge 2$ and $\vmo(\tilde{\psi})\ge 1$ and assume that all the entries of $\phi$ are continuous, then by Theorem~\ref{thm:qp:c} we have $[\ol{\wh{\eta}}^\tp \wh{\tilde{\eta}}]''(0)=
[\ol{\wh{\phi}}^\tp \Theta^\tp \wh{\tilde{\phi}}]''(0)=0$ and
$\{\mathring{\qp}_n\}_{n\in \N}$ must exhibit the Gibbs phenomenon at all points.
According to the results in this paper, the oblique extension principle can increase the vanishing moments of the framelets $\psi$ and $\tilde{\psi}$ as well as the accuracy orders of the quasi-projection operators $\mathring{\qp}_n$, but at the cost of the Gibbs phenomenon.

\medskip

\noindent \textbf{Acknowledgment:} The author would like to thank the reviewers for their valuable suggestions which improved the presentation of the paper.


\begin{thebibliography}{99}

%\bibitem{ak99}
%N.~Atreas and C.~Karanikas, Gibbs phenomenon on sampling series based on Shannon's and Meyer's wavelet analysis.
%\emph{J. Fourier Anal. Appl.} \textbf{5} (1999), 575--588.

\bibitem{chs02}
C.~K.~Chui, W.~He, and J.~St\"ockler, Compactly supported tight and sibling frames with maximum vanishing moments. \emph{Appl. Comput. Harmon. Anal.} \textbf{13} (2002), 224--262.

\bibitem{daubook}
I.~Daubechies, Ten lectures on wavelets. CBMS-NSF Series in Applied Mathematics, \textbf{61}. SIAM, 1992.

%\bibitem{dgm86}
%I.~Daubechies, A.~Grossmann, and Y.~Meyer, Painless nonorthogonal expansions. \emph{J. Math. Phys.} \textbf{27} (1986), 1271--1283.

\bibitem{dh04}
I.~Daubechies and B.~Han, Pairs of dual wavelet frames from any two refinable functions. \emph{Constr. Approx.} \textbf{20} (2004), 325--352.

\bibitem{dhrs03}
I.~Daubechies, B.~Han, A.~Ron, and Z.~Shen, Framelets: MRA-based constructions of wavelet frames. \emph{Appl. Comput. Harmon. Anal.} \textbf{14} (2003), 1--46.

\bibitem{gib}
J.~W.~Gibbs, Letter to the editor, \emph{Nature} \textbf{59}, (1899), 606.

\bibitem{gc95}
S.~M.~Gomes and E.~Cortina, Some results on the convergence of sampling series based on convolution integrals. \emph{SIAM J. Math. Anal.} \textbf{26} (1995), 1386--1402.

\bibitem{gs97}
D.~Gottlieb and C.~Shu, On the Gibbs phenomenon and its resolution. \emph{SIAM Rev.} \textbf{39} (1997), 644--688.

\bibitem{han97}
B.~Han, On dual wavelet tight frames. \emph{Appl. Comput. Harmon. Anal.} \textbf{4} (1997), 380--413.

\bibitem{han09}
B.~Han, Dual multiwavelet frames with high balancing order and compact fast frame transform. \emph{Appl. Comput. Harmon. Anal.} \textbf{26} (2009), 14--42.

\bibitem{han10}
B.~Han, Pairs of frequency-based nonhomogeneous dual wavelet frames in the distribution space. \emph{Appl. Comput. Harmon. Anal.} \textbf{29} (2010), 330--353.

\bibitem{han12}
B.~Han, Nonhomogeneous wavelet systems in high dimensions. \emph{Appl. Comput. Harmon. Anal.} \textbf{32} (2012), 169--196.

\bibitem{han17}
B.~Han, Homogeneous wavelets and framelets with the refinable structure. \emph{Sci. China Math.} \textbf{60} (2017), 2173--2198.

\bibitem{hanbook}
B.~Han, Framelets and wavelets: Algorithms, analysis, and applications. \emph{Applied and Numerical Harmonic Analysis}. Birkh\"auser/Springer, Cham, 2017. xxxiii + 724 pp.

\bibitem{hm03}
B.~Han and Q.~Mo, Multiwavelet frames from refinable function vectors. \emph{Adv. Comput. Math.} \textbf{18} (2003), 211--245.

%\bibitem{jerbook11}
%A.~Jerri, Advances in the Gibbs phenomenon,

\bibitem{jerbook}
A.~Jerri, The Gibbs phenomenon in Fourier analysis, splines and wavelet approximations. Kluwer Academic Publishers, Netherlands, 1998.

\bibitem{jz95}
K.~Jetter and D.-X.~Zhou, Order of linear approximation from shift-invariant spaces. \emph{Constr. Approx.} \textbf{11} (1995), 423--438.


\bibitem{jia04}
R.-Q.~Jia, Approximation with scaled shift-invariant spaces by means of quasi-projection operators. \emph{J. Approx. Theory} \textbf{131} (2004), 30--46.

\bibitem{jj02}
R.-Q.~Jia and Q.-T.~Jiang, Approximation power of refinable vectors of functions. \emph{Wavelet analysis and applications}, 155--178, AMS/IP Stud. Adv. Math., 25, Amer. Math. Soc., Providence, RI, 2002.

\bibitem{kel96}
S.~E.~Kelly, Gibbs phenomenon for wavelets. \emph{Appl. Comput. Harmon. Anal.} \textbf{3} (1996), 72--81.

%\bibitem{mal98}
%S.~Mallat, A wavelet tour of signal processing. Academic Press, Inc., San Diego, CA, 1998. xxiv+577 pp.

\bibitem{ml18}
M.~Mohammad and E.-B. Lin, Gibbs phenomenon in tight framelet expansions. \emph{Commun. Nonlinear Sci. Numer. Simul.} \textbf{55} (2018), 84--92.

\bibitem{rs97}
A.~Ron and Z.~Shen, Affine systems in $L_2(\R^d)$: the analysis of the analysis operator. \emph{J. Funct. Anal.} \textbf{148} (1997), 408--447.

\bibitem{rf05}
D.~K.~Ruch and P.~J.~Van Fleet, Gibbs' phenomenon for nonnegative compactly supported scaling vectors.
\emph{J. Math. Anal. Appl.} \textbf{304} (2005), 370--382.

\bibitem{shen02}
X.~Shen, Gibbs phenomenon in orthogonal wavelet expansion.
\emph{J. Math. Study} \textbf{35} (2002), no. 4, 343--357.

\bibitem{shen11}
X.~Shen, Gibbs phenomenon for orthogonal wavelets with compact support, in \emph{Advances in the Gibbs Phenomenon}, 337--369,
A.~Jerri ed., Sampling Publ., 2011.



\bibitem{sv96}
H.-T.~Shim and H.~Volkmer, On the Gibbs phenomenon for wavelet expansions. \emph{J. Approx. Theory} \textbf{84} (1996), 74--95.

%\bibitem{ws98}
%G.~Walter and H.-T.~Shim, Gibbs' phenomenon for sampling series and what to do about it.
%\emph{J. Fourier Anal. Appl.} \textbf{4} (1998), 357--375.

\bibitem{wil48}
H.~Wilbraham, On a certain periodic function. \emph{Cambridge Dublin Math. J.} \textbf{3} (1848), 198--201.

\end{thebibliography}
\end{document}